%% file: main.tex
\documentclass[a4paper,onecolumn]{IEEEtran}

\input{defs}

\title{Generic Decoding in the Sum-Rank Metric}

\author{\IEEEauthorblockN{Sven Puchinger, \emph{IEEE Member}, Julian Renner \emph{IEEE Student Member}, Johan Rosenkilde}
\thanks{%
This work was partly presented at the IEEE International Symposium on Information Theory (ISIT), 2020 \cite{puchinger2020generic}. %
S.~Puchinger and J.~Rosenkilde are with the Department of Applied Mathematics and Computer Science, Technical University of Denmark (DTU), Denmark.
J.~Renner is with the Institute for Communications Engineering, Technical University of Munich (TUM), Germany. %
S.~Puchinger has been supported by the European Union's Horizon 2020 research and innovation programme under the Marie Sklodowska-Curie grant agreement no.~713683.
J.~Renner has been supported by the European Union's Horizon 2020 research and innovation programme under the European Research Council (ERC) grant agreement no.~801434.
}}

\IEEEoverridecommandlockouts

\begin{document}

\maketitle

\begin{abstract}
We propose the first non-trivial generic decoding algorithm for codes in the sum-rank metric.
The new method combines ideas of well-known generic decoders in the Hamming and rank metric.
For the same code parameters and number of errors, the new generic decoder has a larger expected complexity than the known generic decoders for the Hamming metric and smaller than the known rank-metric decoders.
Furthermore, we give a formal hardness reduction, providing evidence that generic sum-rank decoding is computationally hard.
As a by-product of the above, we solve some fundamental coding problems in the sum-rank metric:
we give an algorithm to compute the exact size of a sphere of a given sum-rank radius, and also give an upper bound as a closed formula; and we study erasure decoding with respect to two different notions of support.
\end{abstract}


\section{Introduction}

The sum-rank \rev{metric} is a family of metrics which contains both Hamming and rank metric as special cases and in general can be seen as a mix of the two.
It was introduced under the name ``extended rank metric'' as a suitable distance measure for multi-shot network coding in 2010 \cite{nobrega2010multishot}.
Since then, several code constructions and efficient decoders have been proposed for the metric \cite{wachter2011partial,wachter2012rank,wachter2015convolutional,napp2017mrd,napp2018faster,martinez2018skew,boucher2019algorithm,martinez2019reliable,caruso2019residues,bartz2020fast,martinezpenas2020sumrank}.
The codes have also been studied in the context of distributed storage \cite{martinez2019universal}, further aspects of network coding \cite{martinez2019reliable}, and space-time codes \cite{shehadeh2020rate}.
Recently, the authors of \cite{byrne2020fundamental} derived several fundamental results on sum-rank-metric codes, including various bounds, MacWilliams identities, and new code constructions.

A generic decoder is an algorithm that takes a code and a received word as input and outputs a codeword that is close to the received word, without any restriction on or knowledge about the structure of the code.
Designing such algorithms has a long tradition in coding theory, both for theoretical and practical reasons: studying the complexity of generic decoding is essential to evaluate the practical security level of code-based cryptosystems such as the McEliece \cite{mceliece1978}, Niederreiter \cite{niederreiter1986knapsack} and Gabidulin--Paramonov--Tretjakov \cite{gabidulin1991ideals} cryptosystems, or the numerous variants thereof.
A trivial generic decoding algorithm is to simply tabulate the input code and compare each codeword with the received word, but there are much more efficient approaches.
For the Hamming metric, the related decision problem is NP-hard \cite{berlekamp1978inherent}, and there is also a hardness reduction for the rank metric \cite{gaborit2016hard}, so it is not surprising that all known generic decoding algorithms have exponential running time in the code parameters.

Prange \cite{prange1962use} presented in 1962 a generic decoder for the Hamming metric whose type is now \rev{known} as information-set decoding.
The basic idea is to repeatedly choose $n-\rev{k}$ random positions, where $n$ is the length and \rev{$k$ the dimension of the code}, until the chosen positions contain all the errors \rev{and the complementary positions form an information set}.
This event can be detected by re-encoding on the remaining \rev{$k$} positions, obtaining a codeword, and seeing that this is close to the received word.
There have been at least $27$ papers improving Prange's algorithm (see the list in \cite[Section~4.1]{bernstein2020classicmceliece}), which have significantly reduced the exponent of the exponential in the complexity expression.

In the rank metric, the first generic decoder was proposed in 1996 \cite{chabaud1996rsd} and since then, there have also been several improvements \cite{ourivski2002rsd, gaborit2016rsd, aragon2018new, bardet2020algebraic}.
One idea here is to repeatedly choose a sub row space (or column space) of the received word until this contains the error row space (resp.\ column space), and when it does use rank-erasure decoding techniques to decode using linear algebra.
The complexity of generic decoding in the rank metric remains significantly higher than in the Hamming metric, which results in a substantial advantage of rank-metric-based cryptosystems over their Hamming-metric analogs.

\subsection{Contributions}
In this paper, we propose the first non-trivial generic decoding algorithm \rev{for arbitrary $\Fqm$-linear codes in the sum-rank metric, where $\Fqm$ denotes the field over which the code is defined. The algorithm takes as input parameters which specify the metric, a parity-check matrix of the code, the received word, and the sum-rank weight of the additive error $t$. The algorithm outputs a vector with weight at most $t$ such that the difference of this vector and the received word is a codeword. If $t$ is at most half the minimum distance of the code, the obtained vector is equal to the error of the received word}.
For this purpose, the algorithm combines the sketched ideas for the Hamming and rank metric: we first randomly choose a rank in each block according to a carefully crafted distribution, and then for each block choose a random row or column space of the given rank.
The process succeeds when the error row or column space in each block is covered, whence decoding is performed using sum-rank erasure decoding using linear algebra.

The most involved part is to design a suitable distribution from which to draw random vectors of a given sum-rank.
In fact, we first observe that even counting the number of such vectors is non-trivial, and so drawing uniformly at random is also non-trivial.
Our distribution is more involved than this, since it turns out that the probability of successful decoding depends on how the rank errors are distributed across blocks.
Roughly, the complexity of our decoding algorithm smoothly interpolates between the basic generic decoders in the two ``extremal'' cases of the sum-rank metric: Hamming and rank metric.

Our work can be seen as a proof-of-concept that known methods of generic decoding can be adapted to the sum-rank metric.
Though out of scope of this paper, it seems reasonable that many improvements for generic decoding in Hamming and rank metric can also be applied, which might further reduce the complexity.

As related results, we study several fundamental problems related to the sum-rank metric:
\begin{itemize}
\item We propose an efficient algorithm to compute the number of vectors of a given sum-rank weight.
Apart from the use in our work, this can e.g.\ be used to efficiently compute the sphere-packing and Gilbert--Varshamov bounds in \cite{byrne2020fundamental}.
\item We give a simple upper bound on the size of a sum-rank-metric sphere.
\item Besides the existing notion of row support \cite{martinez2019theory} and an associated row-erasure decoder \cite{martinez2019universal}, we introduce \rev{a} ``transposed'' notion of column support and an associated column-erasure decoder.
We analyze the computational complexity of both erasure decoders.

\end{itemize}
Finally, we generalize the formal hardness proof of \cite{gaborit2016hard} from the rank metric to the sum-rank metric.
We show that if, for sufficiently large base field, the decisional sum-rank syndrome decoding problem is in the complexity class $\ZPP$, then $\NP = \ZPP$.
Loosely, $\ZPP$ is the set of problems which are computationally easy if one is allowed to use randomness, and includes the problems which are easy to solve deterministically, i.e.~$\mathsf{P}$.
Our result means that sum-rank syndrome decoding is either hard (i.e.~not in $\ZPP$), or that \emph{all} $\mathsf{NP}$ problems are easy.

\subsection{Reader's Guide}

After giving some preliminaires in \cref{sec:preliminaries}, we study the problem of counting vectors of a given sum-rank weight in \cref{sec:enumerating_sr_vectors}. This gives a first comparative line for the generic decoder and is also required for the formal hardness proof.
In \cref{sec:support_and_erasure_decoding}, we introduce two notions of support in the sum-rank metric and show how to efficiently erasure-decode w.r.t.\ these types of support.
Erasure decoding is an essential ingredient of the new generic decoder.
\cref{sec:super_support_finding} presents the generic decoder. We explain how to randomly find a super-support of the error and show how to efficiently implement and bound the complexity of the proposed algorithm.
In \cref{sec:comparison}, we compare the generic decoder to other (naive) generic decoders, as well as existing algorithms for the Hamming and rank metric.
\cref{sec:hardness} presents the formal hardness proof.

\section{Preliminaries}\label{sec:preliminaries}

\subsection{Notation}

Let $q$ be a prime power and $m$ be a positive integer: the codes we consider are over $\Fqm$, the finite field with $q^m$ elements, whose elements we often expand into $\Fq^m$ vectors.
For $r \in \ZZ_{> 0}$ and a fixed basis of $\Fqm$ over $\Fq$, we define the mapping
\begin{align*}
\extFqmFqVector{r} \, : \, \Fqm^r &\mapsto \Fq^{m \times r}, \\
\x &\mapsto \ve{X},
\end{align*}
where the $i$-th column of $\ve{X}$ is the expansion of $x_i$ in the fixed basis over $\Fq$.
We use the big-O notation family to state asymptotic costs of algorithms, and $O^\sim(\cdot)$, which neglects logarithmic terms in the input parameter.
For a finite set $\mathcal{S}$, we denote by $s \sample \mathcal{S}$ the operation of drawing uniformly at random an element $s$ from $\mathcal{S}$.

\subsection{Sum-Rank Metric}

Throughout the paper, $n$ is the length of the studied codes, and $\ell$ is a blocking parameter satisfying $\ell \mid n$.
The length of each block is $\npr := n/\ell$, and we let $\nmmin := \min\{\npr,m\}$.
For a vector $\x \in \Fqm^{\npr}$, we define $\rk_{\Fq}(\x) := \dim_{\Fq} \langle x_1,\dots,x_{\npr}\rangle_{\Fq} = \rk_{\Fq}(\extFqmFqVector{\npr}(\x))$.
Obviously, $\rk_{\Fq}(\x) \leq \nmmin$.
The sum-rank metric is defined as follows.

\begin{definition}
The ($\ell$-)sum-rank weight is defined as
\begin{align*}
  \wtSR \, : \, \Fqm^n &\to \ZZ_{\geq 0}, \\
  \x  &\mapsto \textstyle\sum_{i=1}^{\ell} \rk_{\Fq}(\x_i),
\end{align*}
where we write $\x = \big[ \x_1 | \x_2 | \dots | \x_\ell \big]$ with $\x_i \in \Fqm^{\npr}$.
We call
\begin{equation*}
[\rk_{\Fq}(\x_1),\dots,\rk_{\Fq}(\x_\ell)] \in \{0,\dots,\nmmin\}^\ell
\end{equation*}
the \emph{weight decomposition} of $\x$.
Furthermore, the ($\ell$-)sum-rank distance is defined as
\begin{align*}
\dSR \, : \, \Fqm^n \times \Fqm^n \to \ZZ_{\geq 0}, \quad [\x,\x'] \mapsto \wtSR(\x-\x').
\end{align*}
\end{definition}

The family of sum-rank metrics includes two well-known metrics as extremal cases: For $\ell=1$, it coincides with the rank metric, $\wtR$, and for $\ell=n$, it is the Hamming metric, $\wtH$.
In between, %
we have $\wtR(\x) \leq \wtSR(\x) \leq \min\big\{\nmmin \ell,\ \wtH(\x)\big\}$  for $\x \in \Fqm^n$.

\begin{remark}
Some results in this paper can be generalized in a relatively straightforward way to the sum-rank metric with varying block size (i.e., subblocks of $\x$ are of the form $\x_i \in \Fqm^{\eta_i}$ for positive integers $\eta_1,\dots,\eta_\ell$ with $\sum_{i=1}^{\ell} \eta_i = n$).
We decided to present only the constant block size case ($\eta_i=\eta$ for all $i$) to avoid an even more technical presentation.
\end{remark}

\subsection{Gaussian Binomial and Number of Matrices}

For non-negative integers $a$ and $b$, the Gaussian binomial $\qbinomial{a}{b}$ is defined by the number of $b$-dimensional subspaces of $\Fq^a$.
We have
\begin{equation*}
\qbinomial{a}{b} = \prod_{i=1}^{b} \frac{q^{a-b+i}-1}{q^i-1}
\end{equation*}
and the bounds~\cite{koetter2008coding}
\begin{equation}
q^{(a-b)b} \leq \textstyle\qbinomial{a}{b} \leq \gamma_q q^{(a-b)b}, \label{eq:bounds_gaussian_binomial}
\end{equation}
where
\begin{equation}
\gamma_q := \prod_{i=1}^{\infty} (1-q^{-i})^{-1}. \label{eq:gamma_q}
\end{equation}
Note that $\gamma_q$ is monotonically decreasing in $q$ with a limit of $1$, and e.g.~$\gamma_2 \approx 3.463$, $\gamma_3 \approx 1.785$, and $\gamma_4 \approx 1.452$.
We let $\NM{q}{a,b,i}$ denote the number of $a \times b$ matrices over $\Fq$ of rank exactly $i$, for $0\leq i \leq \min\{a,b\}$.
We have~\cite{migler2004weight}:
\begin{align}
\NM{q}{a,b,i} = \prod_{j=0}^{i-1} \tfrac{(q^a-q^j)(q^{b}-q^j)}{q^i-q^j} \leq 4 q^{i(a+b)-i^2}. \label{eq:num_matrices_of_given_rank}
\end{align}

\subsection{Weight Decompositions and Partitions}

\noindent
For a non-negative integer $t \leq \ell \mu$, we define the set
\begin{align*}
\Tsett := \left\{ \t \in \{0,\dots,\mu\}^\ell \, : \, \sum_{i=1}^{\ell} t_i = t \right\},
\end{align*}
which contains all possible weight decompositions of a vector with $\ell$-sum-rank weight $t$.

The set $\Tsett$ has also a combinatorial interpretation: its elements correspond exactly to the ordered partitions of the integer $t$ with part size at most $\mu$ and number of parts at most $\ell$.
Hence, its cardinality is the $t$-th coefficient of the generating polynomial\footnote{We would like to thank Cornelia Ott for deriving this closed-form expression for $|\Tsett|$.} %
\begin{align*}
p^{(\ell,\mu)}(X) = \left(\sum_{i=0}^{\mu}X^i\right)^\ell,
\end{align*}
i.e.,
\begin{align*}
|\Tsett| = p_t^{(\ell,\mu)} = \sum_{i=0}^{\lfloor \frac{t}{\mu+1} \rfloor} (-1)^i \binom{\ell}{i} \binom{t+\ell-1-(\mu+1)i}{\ell-1}.
\end{align*}
In particular, $|\Tsett|$ can be computed efficiently, and we have the upper bound
\begin{align*}
|\Tsett| \leq \binom{\ell+t-1}{\ell-1}.
\end{align*}
Depending on the relative size of $\ell$ and $\mu$, the cardinality $|\Tsett|$ may grow super-polynomially in $t$.

\subsection{Linear Codes}
\rev{
Throughout this paper, we consider $\Fqm$-linear codes. An $\Fqm$-linear code $\Code$ over $\Fqm$ of dimension $k$ and length $n$ is an $\Fqm$-linear $k$-dimensional subspace of $\Fqm^n$, and we write $\Code[n,k]_{\Fqm}$. The minimum ($\ell$-)sum-rank distance of $\mathcal{C}$ is given by 
\begin{equation*}
d = \min_{\substack{\c,\d \in \mathcal{C} \\ \c \ne \d}} \{\dSR(\c,\d)\} .
\end{equation*}
If $d$ is known, we call the code $\mathcal{C}$ an $[n,k,d]_{\Fqm}$ code. A matrix $\G \in \Fqm^{k\times n}$ is a generator matrix of $\mathcal{C}$ if and only if its rows form a basis of $\mathcal{C}$. Furthermore, a parity-check matrix $\H \in \Fqm^{(n-k)\times n}$ of $\mathcal{C}$ is matrix whose rows form a basis of the right kernel of $\G$.
}

In this paper, we aim at solving the following problem for any given code $\Code$:
\rev{
\begin{problem}[Generic Sum-Rank-Metric Decoding]\label{prob:generic_sr_decoding}\hfill
	\begin{description}
		\item \textbf{Given:}
		\begin{itemize}
			\item Parameters $q,m,k,n,\ell,t$ with $\ell\mid n$ and $0 \leq t \leq \min\{n,m\}\ell$
			\item Parity-check matrix $\H \in \Fqm^{(n-k)\times n}$ of an $\Fqm$-linear $[n,k]_{\Fqm}$ code $\Code$ 
			\item Received vector $\r = \c + \e \in \Fqm^n$, where $\c \in \Code$ and $\wtSR(\e) = t$
		\end{itemize}
              \item \textbf{Objective:} Find a vector $\e'$ with $\wtSR(\e')\leq t$ such that $\r-\e' \in \Code$.
	\end{description}
\end{problem}
}

\rev{
\begin{remark}
	We formulate \cref{prob:generic_sr_decoding} such that the sum-rank weight of the additive error is known and at least one solution to the problem exists. This results from the fact that this is true for most of the applications of generic decoding algorithms. For instance in the code-based encryption schemes BIKE~\cite{aragon2019bike}, HQC~\cite{aguilar2019hqc3rd}, ROLLO~\cite{aguilar2019rollo}, RQC~\cite{aguilar2019rqc}, and ClassicMcEliece~\cite{albrecht2019classicmceliece}, whose security relies on generic decoding in either the Hamming or the rank metric (all systems reached at least the second round of the NIST post-quantum standardization process~\cite{NIST2017post}).
\end{remark}
}

\section{Counting Error Vectors}\label{sec:enumerating_sr_vectors}

As the generic decoding problem can be solved by brute-forcing through all vectors of a given sum-rank weight, we are interested in finding the number of such vectors.
The question of counting is also related to explicitly writing down a list of such vectors (hence, how to realize this naive generic decoder) and provides a comparative line for the complexity of our new generic decoder that we present in the remainder of the paper.
In the extreme cases of the Hamming and rank metric, simple closed-form expressions are easy to obtain.
The question seems more involved for the general sum-rank metric.

We denote by $\mathcal{N}_{q,\npr,m}(t,\ell)$ the number of vectors in $\Fqm^{\npr\ell}$ of $\ell$-sum- rank weight exactly $t\leq \nmmin\ell$.
It is easy to see that we have
\begin{align*}
\mathcal{N}_{q,\npr,m}(t,\ell) = \sum_{\t \in \Tsett} \prod_{i=1}^{\ell} \NM{q}{m,\npr,t_i}.
\end{align*}
However, the number of terms in this formula is $|\Tsett|$ and it is not obvious how the sum can be computed efficiently.
For this reason, we propose an efficient dynamic programming routine to compute the number.
The method is based on the following lemma and outlined in Algorithm~\ref{alg:number_sum-rank-words} below;
note that $q$, $\npr$, and $m$ remain constant throughout the recursion.

\begin{lemma}\label{lem:number_fixed-weight-words_recursion}
  $\mathcal{N}_{q,\npr,m}(t,\ell) \!= 0$ for $t > \nmmin\ell$.
  Otherwise:
\begin{align*}
\mathcal{N}_{q,\npr,m}(t,\ell) =
\begin{cases}
\NM{q}{m,\npr,t}, &\text{if } \ell=1, \\
 \displaystyle \sum_{t'=0}^{\min\{\npr,m,t\}} \hspace{-0.4cm}\NM{q}{m,\npr,t'} \cdot \mathcal{N}_{q,\npr,m}(\rev{t-t'},\ell-1) , &\text{if } \ell>1, \\
\end{cases} \ .
\end{align*}
\end{lemma}

\begin{proof}
  The first claim is obvious since each of the $\ell$ blocks can have at most rank weight $\nmmin$.
  For $\ell=1$, the formula is simply the number of $m \times \npr$ matrices of rank $t$. For larger $\ell$, we sum up over the number of possibilities to choose the rank weight $t'$ of the first block multiplied with the number of sum-rank weight words in the remaining $\ell-1$ blocks.
\end{proof}

We also give a simple upper bound on $\mathcal{N}_{q,\npr,m}(t,\ell)$, which we use for bounding the complexity of \cref{alg:number_sum-rank-words}, as well as for proving the formal hardness of generic decoding in \cref{sec:hardness}.

\begin{theorem}\label{thm:bound_number_fixed-weight-words}
  For $\ell>1$ and $t \leq \nmmin\ell$, the number of vectors in $\Fqm^{\npr\ell}$ of $\ell$-sum rank weight $t$ can be bounded by
  \begin{equation*}
 \mathcal{N}_{q,\npr,m}(t,\ell)\leq \gamma_q^{\ell} \binom{\ell+t-1}{\ell-1} q^{t(m+\npr-\frac{t}{\ell})},
  \end{equation*}
  where $\gamma_q \leq 3.5$ is given in \eqref{eq:gamma_q}.
\end{theorem}

\begin{proof}
  By definition,
  \begin{align*}
  \mathcal{N}_{q,\npr,m}(t,\ell) &= \sum_{\t \in \Tsett} \prod_{i=1}^{\ell} \NM{q}{m,\npr,t_i} \\
                     &\leq |\Tsett| \max_{\t \in \Tsett} \Bigg \{  \prod_{i=1}^{\ell} \NM{q}{m,\npr,t_i} \Bigg \} \\
                     &\leq \binom{\ell+t-1}{\ell-1} \gamma_q^{\ell} q^{\max_{\t \in \Tsett} \big \{ \sum_{i=1}^{\ell} t_i (m+\eta-t_i) \big\} },
  \end{align*}
  where the latter inequality follows from $|\Tsett| \leq \binom{\ell+t-1}{\ell-1}$ and  $\NM{q}{m,\npr,t_i} \leq \gamma_q q^{t_i(m+\npr-t_i)}$.
  Thus we should upper-bound $\max_{\t \in \Tsett} \big \{ \sum_{i=1}^{\ell} t_i (m+\eta-t_i) \big\}$ subject to $\sum_{i=1}^{\ell} t_i = t$, which simplifies to maximising
  \[
    t(m+\eta) - \sum_{i=1}^{\ell} t_i^2 \ .
  \]
  By Jensen's inequality, this is upper-bounded by choosing $t_i = t/\ell$ for all $i$, i.e.
  \[
    \max_{\t \in \Tsett} \left \{ \sum_{i=1}^{\ell} t_i (m+\eta-t_i) \right\} \leq t(m+\eta) - t^2/\ell \ .
  \]
\end{proof}

\cref{fig:number_of_errors_of_given_weight_example} shows example values of $\mathcal{N}_{q,\npr,m}(t,\ell)$ and the bound in \cref{thm:bound_number_fixed-weight-words} for different divisors $\ell$ of a fixed length $n$.
It seems that the bound is quite tight for most values of $\ell$, and only significantly differs for $\ell$ close to $n$.
This deviation is due to the factor $\gamma_q^{\ell}$, which is large for these values of $\ell$, and which is due to a relatively bad bound on the number of matrices.
Note that for $\ell=n$, we know better bounds on $\mathcal{N}_{q,\npr,m}(t,\ell)$ from the Hamming metric.

\begin{figure}[ht!]
\begin{center}
\includegraphics{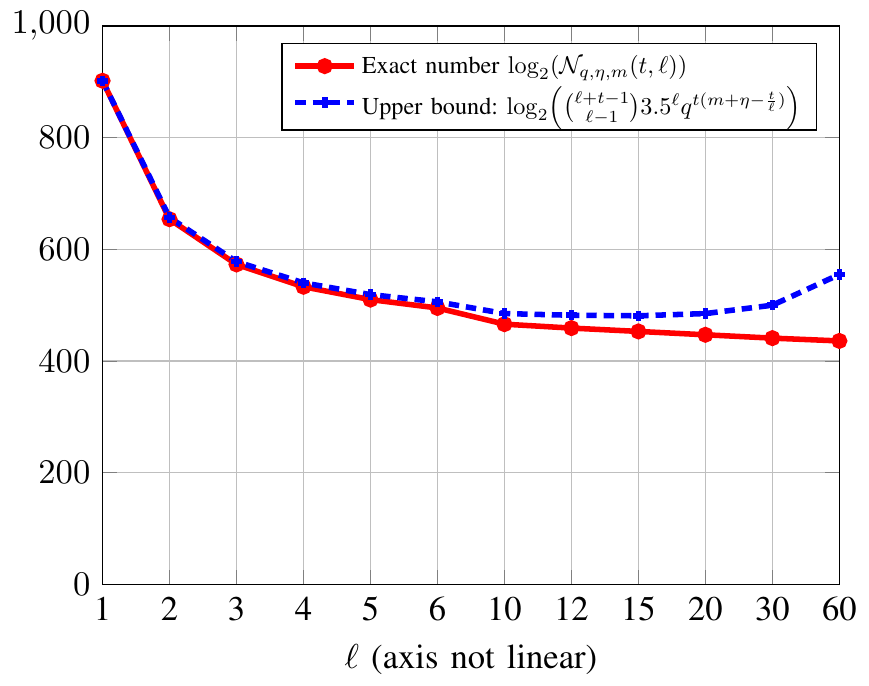}
\end{center}
\vspace{-0.5cm}
\caption{Comparison of the exact number of vectors of sum-rank weight $t=10$ and the derived upper bound for $q=2$, $m=40$, $n=60$ as a function of $\ell$.}
\label{fig:number_of_errors_of_given_weight_example}
\end{figure}

\begin{algorithm}
	\caption{$\textsf{Compute } \mathcal{N}_{q,\npr,m}(t,\ell)$}\label{alg:number_sum-rank-words}
	\SetKwInOut{Input}{Input}\SetKwInOut{Output}{Output}
	\Input{Prime power $q$ and $\npr, m, \ell, t \in \ZZ_{\geq 0}$ such that $0 < t \leq \nmmin \ell$ and $\nmmin := \min\{\npr,m\}$}
	\Output{Number $\mathcal{N}_{q,\npr,m}(t,\ell)$ of vectors in $\Fqm^{\npr \ell}$ of $\ell$-sum-rank weight $t$} %
	Initialize table of integers $\{\mathsf{N}(t',\ell') = 0 \}_{t'=0,\dots,t}^{\ell'=1,\dots,\ell}$ \\
	\For{$t'=0,\dots,t$}{
		$\mathsf{N}(t',1) \gets \NM{q}{m,\npr,t'}$
	}
	\For{$\ell'=2,\dots,\ell$}{
		\For{$t'=0,\dots,t$}{
			$\mathsf{N}(t',\ell') \gets \hspace{-0.2cm}\displaystyle \sum_{t''=\rev{0}}^{\min\{\nmmin,t'\}} \hspace{-0.1cm} \NM{q}{m,\npr,t''} \mathsf{N}(\rev{t-t''},\ell'-1)$
		}
	}
	\Return{$\mathsf{N}(t,\ell)$}
\end{algorithm}

\begin{theorem}\label{thm:complexity_counting}
Algorithm~\ref{alg:number_sum-rank-words} is correct and has bit complexity
\begin{align*}
O^\sim\big( \ell^2 t^2 + \ell t^3(m+\npr)\log(q) \big).
\end{align*}
\end{theorem}

\begin{proof}
The algorithm computes a table that fulfills $\mathsf{N}(t',\ell') = \mathcal{N}_{q,\npr,m}(t',\ell')$ for all $t'=0,\dots,t$ and $\ell'=1,\dots,\ell$ using the recursive formula in Lemma~\ref{lem:number_fixed-weight-words_recursion}. This implies the correctness.

Complexity-wise, the algorithm performs $\ell t^2$ integer multiplications, where the size of the integers are such that they impact performance.
An upper bound is given by
\begin{align}
\mathcal{N}_{q,\npr,m}(t,\ell) & \leq \binom{\ell+t-1}{\ell-1} \gamma_q^{\ell} q^{t(m+\npr-\frac{t}{\ell})}\notag \\
&\rev{\leq \left(e\tfrac{\ell+t-1}{\ell-1}\right)^{\ell-1}\gamma_q^{\ell} q^{t(m+\npr-\frac{t}{\ell})} }
\label{eq:simple_bound_sphere_size}
\end{align}
where the first inequality follows from \cref{thm:bound_number_fixed-weight-words}, \rev{the second inequality follows from an upper bound on binomial coefficients, and $e$ is Euler's constant}.
Since integer multiplication can be implemented with quasi-linear bit operations in the bit size of the involved integers \cite{harvey2019faster}, each multiplication costs at most
\begin{align*}
  &O^\sim\left( (\ell-1) \log\!\left( e \gamma_q \frac{\ell+t-1}{\ell-1} \right) + t(m+\npr-\tfrac{t}{\ell})\log(q) \right)\\
  \subseteq ~&O^\sim\big( \ell + t(m+\npr-\tfrac{t}{\ell})\log(q) \big) \\
  \subseteq ~&O^\sim\big( \ell + t(m+\npr)\log(q) \ . \qedhere
\end{align*}

\end{proof}

\begin{corollary}\label{cor:naive_generic_decoder_Werror}
There is a deterministic algorithm that solves Problem~\ref{prob:generic_sr_decoding} using at most $\Werror$ operations in $\Fq$, where
\begin{equation}
\Werror \in  O\left(n(n-k)m^2 \tbinom{\ell+t-1}{\ell-1} \gamma_q^{\ell} q^{t(m+\npr-\frac{t}{\ell})}\right) \ . \label{eq:Werror}
\end{equation}
\end{corollary}
\begin{proof}
  Algorithm~\ref{alg:number_sum-rank-words} can be easily adapted to create a list of all errors of sum-rank weight $t$: instead of storing the number of vectors in the table $\mathsf{N}(\cdot, \cdot)$, we store lists of the respective vectors. By brute-forcing the overall list and checking whether the received word minus each error is a codeword (this costs at most $O(n(n-k)m^2)$ operations over $\Fq$. \rev{Notably, the constant in the $O$ notation is small.}), we obtain a generic decoder with complexity
  \begin{align*}
     &O\!\left(n(n-k)m^2 \mathcal{N}_{q,\npr,m}(t,\ell)\right) \\
    \leq~ &O\!\left(n(n-k)m^2 \tbinom{\ell+t-1}{\ell-1} \gamma_q^{\ell} q^{t(m+\npr-\frac{t}{\ell})}\right),
  \end{align*}
  using \eqref{eq:simple_bound_sphere_size} in the proof of Theorem~\ref{thm:complexity_counting}.
\end{proof}
The binomial in the expression can be simplified, depending on the relation between $t$ and $\ell$: for instance, since $t \leq \ell \nmmin$, then $\frac t {\ell-1} \leq 2 \mu$, and therefore
\[
  \tbinom{\ell+t-1}{\ell-1}
  \leq \left(e\tfrac{\ell+t-1}{\ell-1}\right)^{\ell-1}
  \in O\!\left( \big[e(2\nmmin+1)\big]^{\ell} \right) \ ,
\]
where $e$ is Euler's constant.

\begin{remark}
\rev{The recursion in \cref{lem:number_fixed-weight-words_recursion} can be turned into an efficient algorithm to draw uniformly at random from the set of vectors of sum-rank weight $t$, see \cref{app:drawerrors}.}
\end{remark}

\begin{remark}\label{rem:further_applications_sphere_size}
In \cite{byrne2020fundamental}, several fundamental bounds of sum-rank-metric codes are derived.
To evaluate two of their bounds, the sphere-packing and Gilbert--Varshamov bound, one needs to efficiently compute the volume of a ball of given sum-rank radius, but this is not addressed in \cite{byrne2020fundamental}.
Algorithm~\ref{alg:number_sum-rank-words} (and a straightforward variant thereof for variable block size and extension degree in each block) provides an efficient method to do this.
Furthermore, the upper bound in \cref{thm:bound_number_fixed-weight-words} allows a significant simplification of their Gilbert--Varshamov bound, though we have not investigated how much weaker it becomes.
\end{remark}

\section{Erasure Decoding and Support in the Sum-Rank Metric}\label{sec:support_and_erasure_decoding}

In Section~\ref{sec:super_support_finding}, we will present a new generic decoding algorithm for the sum-rank metric.
The idea is similar to the generic decoders in the Hamming and rank metric: first we find the ``support'' of an error (e.g., the error positions in the Hamming metric) in a randomized fashion and second we compute the full error by erasure decoding (e.g., computing the error values after having found the error positions).

In this section, we therefore study two notions of support in the sum-rank metric:
row and column support. We describe erasure decoding w.r.t.\ these two notions, i.e., we explain under which conditions and in which complexity we can uniquely recover an error from a received word given its support.
We will see in the next section that the two notions of support are advantageous on different parameters: If $\eta \leq m$, our generic decoder is faster if we aim at finding a row support, and for $\eta \geq m$, it is faster to find a column support.

The notion of row support was already introduced in \cite{martinez2019theory} in a different context.
From \cite[Corollary~1]{martinez2019universal}, one can easily derive that erasure decoding w.r.t.\ this support is unique if the support weight is smaller than the minimum distance.
For the row support, our contributions are hence an explicit description of an erasure decoder and a complexity bound.
We are not aware of previous work on the column support or erasure decoding thereof.

\subsection{Two Notions of Support}\label{ssec:support}

The following lemma gives rise to two notions of ``support'' in the sum-rank metric, which we state in Definition~\ref{def:support} below.

\begin{lemma}\label{lem:error_decomposition}
Let $\e \in \Fqm^n$ have $\ell$-sum-rank weight $t$ and let $\t$ be its weight decomposition.
Then there are vectors
\begin{equation*}
\a_i \in \Fqm^{t_i}, \, \rk_{\Fq}\!\left(\a_i\right) = t_i, \quad \text{for } i=1,\dots,\ell,
\end{equation*}
as well as matrices over the sub-field $\Fq$:
\begin{equation*}
\B_i \in \Fq^{t_i \times \npr}, \rk_{\Fq} \!\left(\B_i\right) = t_i, \quad  \text{for } i=1,\dots,\ell,
\end{equation*}
such that \vspace{-0.4cm}
\begin{equation*}
\e = \overbrace{\begin{bmatrix}
\a_1 & \a_2 & \a_3 & \dots & \a_\ell
\end{bmatrix}}^{=: \, \a \, \in \,\Fqm^t}
\cdot
\overbrace{\begin{bmatrix}
\B_1 & \0 & \0 & \dots & \0 \\
\0 & \B_2 & \0 & \dots & \0 \\
\0 & \0 & \B_3 & \dots & \0 \\
\0 & \0 & \0 & \dots & \B_\ell \\
\end{bmatrix}}^{=: \, \B \, \in \, \Fq^{t \times n}}.
\end{equation*}
Furthermore, the decomposition is unique up to elementary $\Fq$-row operations on the matrices $\B_i$. In particular, the $\Fq$-row spaces of the matrices $\B_i$, as well as the $\Fq$-column space of $\extFqmFqVector{t_i}(\a_i)$, are uniquely determined by $\e$.
\end{lemma}

\begin{proof}
  By basic linear algebra, see e.g.~\cite{Gabidulin_TheoryOfCodes_1985}, there is an $\a_i \in \Fqm^{t_i}$ and $\B_i \in \Fq^{t_i \times \npr}$ such that $\e_i = \a_i\B_i$.
  Also the uniqueness up to row operations follows directly from the analogous results in the rank metric.
\end{proof}

\begin{definition}\label{def:support}
Let $\e \in \Fqm^n$ be of sum-rank weight $t$.
\begin{itemize}
\item \textbf{Row Support:} The row support of $\e$ is defined as the product of subspaces
\begin{equation*}
\Eset^{\mathsf{(R)}}_{\e} := \Eset_1^{\mathsf{(R)}} \times \Eset_2^{\mathsf{(R)}} \times \dots \times \Eset_\ell^{\mathsf{(R)}},
\end{equation*}
where $\Eset_i^{\mathsf{(R)}} \subseteq \Fq^\eta$ is the $\Fq$-row space of $\B_i \in \Fq^{t_i \times \eta}$ as in \cref{lem:error_decomposition}.
A product
\begin{equation*}
\Fset^{\mathsf{(R)}} := \Fset_1^{\mathsf{(R)}} \times \Fset_2^{\mathsf{(R)}} \times \dots \times \Fset_\ell^{\mathsf{(R)}}
\end{equation*}
of subspaces $\Fset_i^{\mathsf{(R)}} \subseteq \Fq^{\npr}$ is called a row super-support of $\e$, denoted by $\Eset^{\mathsf{(R)}}_{\e} \subseteq \Fset^{\mathsf{(R)}}$, if $\Eset_i^{\mathsf{(R)}} \subseteq \Fset_i^{\mathsf{(R)}}$ for all $i$.
\item \textbf{Column Support:} The column support of $\e$ is defined by
\begin{equation*}
\Eset^{\mathsf{(C)}}_{\e} := \Eset_1^{\mathsf{(C)}} \times \Eset_2^{\mathsf{(C)}} \times \dots \times \Eset_\ell^{\mathsf{(C)}},
\end{equation*}
where $\Eset_i^{\mathsf{(C)}} \subseteq \Fq^m$ is the column space of $\extFqmFqVector{t_i}(\a_i) \in \Fq^{m \times t_i}$ as in \cref{lem:error_decomposition}.
A column super-support $\Fset^{\mathsf{(C)}} \supseteq \Eset^{\mathsf{(C)}}_{\e}$ is defined analogously to the row case.
\end{itemize}
If it is clear from the context that we mean the row or column support, we will simply write $\Eset_{\e}$, $\Fset$, and $\Eset_{\e} \subseteq \Fset$, and omit the prefixes ``row'' and ``column'' to simplify notation.
\end{definition}

\begin{remark}
It is easily seen that Definition~\ref{def:support} specializes the usual notions of support for Hamming metric when $\ell=n$, and the row and column support, respectively, in the rank metric for $\ell=1$.
\end{remark}

The following notation will be useful in the next section.

\begin{definition}
Let $\ambdim$ be a positive integer and $0 \leq s \leq \ell \ambdim$.
For $\s \in \TsetWITHARGUMENTS{s, \ell, \ambdim}$, we define the set
\begin{align*}
\SupportSet(\s) := \Big\{ \Fset = \Fset_1 \times \cdots \times \Fset_\ell \, : \, \Fset_i \text{ is an $s_i$-dimensional subspace of $\Fq^\ambdim$} \Big\}.
\end{align*}
For any $\Fset \in \SupportSet(\s)$, we say that its weight decomposition is $\s$ and its weight is $s$.
\end{definition}

\subsection{Erasure Decoding}\label{ssec:erasure_decoding}

The following theorem generalizes the classical Hamming metric statement that $d-1$ is the maximal number of linearly independent columns, as well as the analogous statement in rank metric \cite[Theorem~1]{Gabidulin_TheoryOfCodes_1985}:

\begin{lemma}\label{lem:characterization_sum-rank-distance_codes}
Let $\H \in \Fqm^{\rev{(n-k)} \times n}$ be a parity-check matrix of a code $\Code[n,k]_{\Fqm}$.
Define for any integer $0\leq t \leq n$ the set
\begin{align*}
\mathcal{B}_{\ell,t} := \left\{ \B  = \left[\begin{matrix}
\B_1 & \0 & \0 & \dots & \0 \\
\0 & \B_2 & \0 & \dots & \0 \\
\0 & \0 & \B_3 & \dots & \0 \\
\0 & \0 & \0 & \dots & \B_\ell \\
\end{matrix}\right] \in \Fq^{t \times n} \, :  \B_i \in \Fq^{t_i \times (n/\ell)}, \, \rk \!\left(\B_i\right) = t_i, \, \sum_{i=1}^\ell t_i = t \right\}&
\end{align*}
Then, $\Code$ has minimum $\ell$-sum-rank distance $d$ if and only if
\begin{itemize}
\item we have $\rk_{\Fqm}\!\big(\H \B^\top \big) = d-1$ for any $\B \in \mathcal{B}_{\ell,d-1}$ and
\item we have $\rk_{\Fqm}\!\big(\H \B^\top \big) <d$ for at least one $\B \in \mathcal{B}_{\ell,d}$.
\end{itemize}
\end{lemma}

\begin{proof}
The proof follows by the decomposition of words of a given $\ell$-sum-rank weight in Lemma~\ref{lem:error_decomposition}, together with the definition of the minimum sum-rank distance, i.e., that $\H \x^\top \neq \0$ for any word of $\wtSR(\x)=d-1$ and there is at least one $\x \in \Fqm^n$ with $\wtSR(\x)=d$ and $\H \x^\top = \0$.
\end{proof}

\cref{lem:characterization_sum-rank-distance_codes} implies the following statement about erasure decoding w.r.t.\ the row support in the sum-rank metric. The uniqueness of the recovered codeword can also be derived from \cite[Corollary~1]{martinez2019universal}.

\begin{theorem}[Column Erasure Decoding]\label{thm:column_erasure_decoding}
Let $\r = \c+\e \in \Fqm^n$ be a received word, where $\c$ is an unknown codeword of a code with minimum sum-rank distance $d$ and $\e$ is an unknown error of sum-rank weight at most $d-1$. If we know a row super-support $\Fset = \Fset^{\mathsf{(C)}}$ of $\e$ of weight at most $d-1$, then we can uniquely recover $\c$ from $\r$ with complexity $O((n-k)^3m^2)$ operations over $\Fq$.
\end{theorem}

\begin{proof}
It follows from Lemma~\ref{lem:error_decomposition} that $\e$ can be written as $\a \B$, where $\B$ is a block-diagonal matrix containing bases of the super-support entries $\Fset_i$. Let $\H$ be a parity-check matrix of \rev{the given code $\mathcal{C}$ of minimum sum-rank distance $d$}. Since $\Fset$ has weight $t \leq d-1$, by Lemma~\ref{lem:characterization_sum-rank-distance_codes}, the matrix $\H \B^\top \in \Fqm^{(n-k) \times t}$ has $\Fqm$-rank $t$.
Hence, the linear system
\begin{equation*}
\H \r^\top = \H \e^\top = (\H \B^\top) \a^\top,
\end{equation*}
where $\a$ is unknown, and $\r$, $\H$, and $\B$ are known, has a unique solution $\a$ and we can uniquely determine $\a$, $\e$, and thus $\c$ using linear-algebraic operations.
Using elementary matrix multiplication, Gaussian elimination, and polynomial multiplication algorithms, the involved operations have the following complexities:
Multiplying $\H \B^\top$ costs $O((n-k)s\npr m)$ operations in $\Fq$ since each row of $\B$ has at most $\npr$ non-zero entries.
The only remaining step is solving the linear system $\big(\H \B^\top\big) \a^\top = \s^\top$, where $\s$ is the syndrome of the received word.
This costs $O(s^2 (n-k))$ operations over $\Fqm$, and any operation in $\Fqm$ costs again $O(m^2)$ operations in $\Fq$.
\end{proof}

Similarly, we can recover a codeword from the received word and a column super-support of the error.

\begin{theorem}[Row Erasure Decoding]\label{thm:row_erasure_decoding}
 Let $\r = \c+\e \in \Fqm^n$ be a received word, where $\c$ is an unknown codeword of a code $\mathcal{C}$ with minimum distance $d$ and parity check matrix $\H \in \Fqm^{(n-k) \times n}$. Further $\e$ is an unknown error of sum-rank weight $t <d$. If we know a column super-support of dimension $t' \leq d-1$,
  then we can uniquely recover $\c$ from $\r$ with complexity $O((n-k)^3m^3)$ in operations over $\Fq$
\end{theorem}
\begin{proof}
  Let $\H= [\H_1,\hdots,\H_{\ell}]$, where $\H_{i}\in \Fqm^{(n-k)\times \npr}$. Then, using the same notation as in Theorem~\ref{thm:column_erasure_decoding}, the syndrome is equal to
  \begin{equation*}
    \s^{\top} = \H \B^{\top} \a^{\top} = \sum_{i=1}^{\ell} \H_{i} \B_i^{\top} \a_i^{\top} = \sum_{i=1}^{\ell} \H_{i} \hat{\B}_i^{\top} \hat{\a}_i^{\top},
  \end{equation*}
where $\hat{\a} = [\hat{\a}_1,\hdots, \hat{\a}_{\ell}] \in \Fqm^{t'}$ is a basis of the known column super-support (more precisely, the columns of $\extFqmFqVector{t_i}(\hat{\a}_i)$ form a basis of the $i$-th constituent subspace of the super-support) of the error and $\hat{\B}_i \in \Fq^{t'_i \times \eta}$.
To perform erasure decoding, we solve the latter system of equations for the $\npr t'$ unknown entries of $\hat{\B}_1,\hdots,\hat{\B}_{\ell}$ over $\Fq$. The system over $\Fq$ can be written as
\begin{equation*}
  \s^{\top}_{\text{ext}} = \hat{\H}_{\text{ext}} \hat{\b}{}^{\top},
\end{equation*}
where $\s_{\text{ext}} \in \Fq^{(n-k)m}$ is the expanded syndrome and the matrix $\hat{\H}_{\text{ext}} \in \Fq^{m(n-k)\times\npr t'}$ depends only on $\H$ and $\hat{\a}$. Further, the vector $\hat{\b}$ is defined as
\begin{equation*}
 \hat{\b} :=
   [ \hat{B}_{111},
    \hdots,
    \hat{B}_{\ell t_{\ell} \npr} ],
\end{equation*}
where $\hat{B}_{i j r}$ denotes the entry in the $j$-th row and $r$-th column of the matrix $\hat{\B}_i$.

The system has a unique solution if and only if $\rk(\hat{\H}_{\text{ext}})= \npr t' $. To see that this is always the case, suppose $\s^{\top}_{\text{ext}} = \hat{\H}_{\text{ext}} \hat{\b}{}^{\top}=\0$ and $\rk(\hat{\H}_{\text{ext}}) < \npr t' $. Then, there exists a vector $\hat{\b} \neq \0$ such that $\hat{\H}_{\text{ext}} \hat{\b}{}^{\top}= \H (\hat{\a} \hat{\B})^{\top} =\0$ which means $\hat{\a} \hat{\B} \in \mathcal{C}\setminus \{\0\}$. Since $\wtSR(\hat{\a} \hat{\B}) = t' < d$, this is a contradiction.

The heaviest step is to solve an $m(n-k) \times \eta t'$ linear system over $\Fq$, where $\eta t' \leq m(n-k)$.
This can be done in $O(m^3(n-k)^3)$ operations over $\Fq$.
\end{proof}

\rev{
\begin{remark}
As we consider $\Fqm$-linear codes in this paper, it is necessary to treat row and column sum-rank supports separately. However, in case $\Fq$-linear or non-linear codes are considered, this distinction can be neglected since transposition preserves $\Fq$-linearity, and therefore, the column support can be thought of as the row support, and vice versa. Note that the presented algorithm can be adapted to $\Fq$-linear codes if an erasure decoder of this code is known. However, deriving an erasure decoder for $\Fq$-linear codes is outside the scope of this paper.
\end{remark}
}

\section{The Generic Decoder}\label{sec:super_support_finding}

We have seen in the previous section that we can uniquely recover an error $\e$ if we find a row or column super-support $\Fset \supseteq \Esete$ of sum-rank weight $s$ with $t \leq s < d$.
In this section, we describe a Las Vegas-type algorithm (\cref{alg:generic_decoder} below) that chooses row or column supports $\Fset$ of weight $s$ at random according to a designed probability mass function (here denoted by $\mathsf{DrawRandomSupport}(s,t,\ambdim)$, see \cref{alg:draw_random_support} in \cref{ssec:support-finding_algorithm}).
Notation-wise, there is no difference between drawing random row or column supports if we allow the ambient space dimension $\dim \Fset_i s $ of a constituent support subspace to be arbitrary. We denote this dimension by $\ambdim$ and set it $\ambdim = \eta$ (i.e., $\Fset_i= \Fset_i^{\mathsf{(R)}} \subseteq \Fq^\eta$) in the row support case and $\ambdim = m$ (i.e., $\Fset_i= \Fset_i^{\mathsf{(C)}} \subseteq \Fq^\eta$) in the column support case.
We also omit the prefixes ``row'' or ``column'' in this section.
This allows us to treat both cases in a unified manner.

\begin{algorithm}[ht!]
	\caption{$\textsf{Generic Sum-Rank Decoder}$}\label{alg:generic_decoder}
	\SetKwInOut{Input}{Input}\SetKwInOut{Output}{Output}
	\Input{\rev{Parameters $q,m,k,n,\ell,t$  \\
               Parity-check matrix $\H \in \Fqm^{(n-k)\times n}$ of an $\Fqm$-linear $[n,k]_{\Fqm}$ code $\Code$ \\
               Received vector $\r \in \Fqm^n$ \\
               Integer $s$ with $t \leq s \leq n-k$}}
	\Output{\rev{Vector $\e' \in \Fqm^n$ such that $\wtSR(\e')\leq t$ and $\r-\e' \in \Code$}}
	$\e' \gets \0$ \\
        $\npr \gets n/\ell$ \\
	$\ambdim \gets \min\{m,\npr\}$ \\
	\While{$\H(\r-\e')^\top \neq \0$ or $\wtSR(\e') \rev{ > } t$}{
		$\Fset \gets \mathsf{DrawRandomSupport}(s,t,\ambdim)$ (\cref{alg:draw_random_support} in \cref{ssec:support-finding_algorithm}) \label{line:generic_decoder_start_iteration} \\
		\If{$\ambdim=\npr$} {
			$\e' \gets$ column erasure decoding w.r.t.\ $\Fset$, $\H$, $\r$ (cf.~\cref{thm:column_erasure_decoding})
		} \Else {
			$\e' \gets$ row erasure decoding w.r.t.\ $\Fset$, $\H$, $\r$ (cf.~\cref{thm:row_erasure_decoding}) \label{line:generic_decoder_stop_iteration}
		}
	}
	\Return{$\e'$}
\end{algorithm}

The main statement of this section is \cref{thm:generic_decoder_simple}, which bounds the expected runtime of  \cref{alg:generic_decoder}. Note that by ignoring the cost of one iteration (i.e., setting $\Wperiteration=1$) in \cref{thm:generic_decoder_simple}, one obtains lower and upper bounds on the expected number of iterations that the algorithm takes until a suitable support is found.
Since the proof is rather technical, we prove it in the course of this section. In the statement, we use the notation $\Q$, which is defined in \eqref{eq:Q} below.

\begin{theorem}\label{thm:generic_decoder_simple}
Let $\c$ be a codeword of a sum-rank metric code $\Code$ of minimum sum-rank distance $d$.
Further, let $\e$ be an error of sum-rank weight $t<d$.
Then, Algorithm~\ref{alg:generic_decoder} with input $\r=\c+\e$ and parameter $s$ with $t \leq s <d$ returns an error $\e'$ of sum-rank weight $t$ such that $\r-\e'$ is a codeword.

Each iteration (Lines~\ref{line:generic_decoder_start_iteration}--\ref{line:generic_decoder_stop_iteration}) of \cref{alg:generic_decoder} costs $\Wperiteration \in O^\sim\!\left(n^3 m^3 \log_2(q)\right)$ bit operations.
By including also the expected number of iterations, we can bound the overall expected runtime (in bit operations) $\Wnew$ of \cref{alg:generic_decoder} by
\begin{align*}
\WnewLB \leq \Wnew \leq \WnewUB \leq \WnewUBsimple,
\end{align*} 
where, for $\ambdim = \mu = \min\{\eta,m\}$, we define (see \eqref{eq:Q} for $\Q$)
\begin{align}
\WnewLB &:= |\Tset|^{-1} \Q,
 \label{eq:WF_good_lower_bound} \\
\WnewUB &:= \Wperiteration \Q 
\text{ and} \label{eq:WF_good_upper_bound} \\
\WnewUBsimple &:= \Wperiteration \tbinom{\ell+t-1}{\ell-1}\gamma_q^\ell q^{t(\ambdim-\frac{s}{\ell})}, \label{eq:WF_simplified_upper_bound}
\end{align}
Furthermore, the more precise bounds \eqref{eq:WF_good_lower_bound} and \eqref{eq:WF_good_upper_bound} can be computed in bit complexity $O^\sim\!\left(tsn^3\mu \ambdim^2 \log_2(q)\right)$. %
\end{theorem}

\begin{proof}
See \cref{ssec:proof_main_statement}.
\end{proof}

\begin{remark}\label{rem:parameter_s=n-k}
We can guarantee uniqueness of erasure decoding in Algorithm~\ref{alg:generic_decoder} only for $s<d$, but it might work up to $s=\min\!\left\{n-k, \lfloor\tfrac{m}{\npr}(n-k)\rfloor\right\}$, depending on the chosen super-support.
Most generic Hamming- and rank-metric decoding papers use $s=n-k$ without analyzing the erasure decoding success probability.
Since in practice, the latter probability is high for many codes, $s= \min\big\{n-k, \lfloor\tfrac{m}{\npr}(n-k)\rfloor\big\}$ is indeed a good heuristic choice for a practical generic decoder.
\end{remark}

\subsection{Aim and Design of the Support Drawing Algorithm}

\rev{Our aim in designing the probability distribution for drawing a random support $\Fset$ of weight $s$ is to minimize the worst-case expected number of iterations until we find a super-support of $\Esete$. Since we draw random supports $\Fset$ until one of them is a super-support of $\Esete$, the expected number of required draws is equal to the inverse of the probability that $\Fset$ contains $\Esete$. As we draw one support $\Fset$ per iteration, we have
\begin{align*}
\NumItWorstCase = \max_{\substack{\e \in \Fqm^n \, : \\ \wtSR(\e) = t}} \left\{\frac{1}{\Pr(\Esete \subseteq \Fset)} \right\}.
\end{align*}
}

Our algorithm draws $\Fset$ in two steps: First, we choose at random a weight decomposition $\s \in \Tsets$ of weight $s$, according to a designed probability distribution $\ps$.
Then, we draw the support $\Fset$ uniformly at random from the set $\SupportSet(\s)$ of supports with weight decomposition $\s$.
The following lemma states that the success probability of this decoder, conditioned on a specific weight decomposition $\s$, only depends on $\s$ and the weight decomposition $\te$ of the error.

\begin{lemma}\label{lem:success_fixed_s_vector}
Let $\e$ be of $\ell$-sum-rank weight $t$.
Further, let $\s \in \Tsets$ and choose $\Fset$ uniformly at random from $\SupportSet(\s)$. %
Then,
\begin{align*}
\rev{\Pr(\Esete \subseteq \Fset \mid \s)} = \rhoprob(\s,\te),
\end{align*}
where we define
\begin{align}
\rhoprob(\s,\t) := \prod_{i=1}^{\ell} \frac{\qbinomial{s_i}{t_i}}{\qbinomial{\ambdim}{t_i}}. \label{eq:definition_varrho}
\end{align}
In particular, \rev{$\Pr(\Esete \subseteq \Fset\mid \s)$} only depends on the decompositions $\s$ and $\te$, and we have $\rev{\Pr(\Esete \subseteq \Fset\mid\s)}>0$ if and only if $\s \succeq \te$\, \rev{where $\succeq$ is the partial order given by coordinate-wise comparisons.}

Furthermore, we have (with $1 \leq \gamma_q \leq 3.5$ as defined in \eqref{eq:gamma_q})
\begin{align}
\gamma_q^{-\ell} q^{-\sum_{i=1}^{\ell} t_i(\ambdim-s_i)} \leq \rhoprob(\s,\t) \leq \gamma_q^\ell q^{-\sum_{i=1}^{\ell} t_i(\ambdim-s_i)}. \label{eq:bounds_rho_s_t}
\end{align}
\end{lemma}

\begin{proof}
Since $\Fset$ is drawn uniformly, the subspaces $\Fset_i$ are drawn independently and uniformly from the set of $s_i$-dimensional subspaces of $\Fq^\ambdim$. Hence, \rev{$\Pr(\Esete \subseteq \Fset\mid \s)$} equals the product of the probabilities that the $i$-th subspace $\Fset_i$ is a superspace of $\Eset_i$. This probability is given by $\qbinomial{\ambdim-t_i}{s_i-t_i}\qbinomial{\ambdim}{s_i}^{-1}$, where the numerator counts the number of possibilities to expand the $t_i$-dimensional subspace $\Eset_i$ into an $s_i$-dimensional space and the denominator gives the total number of $s_i$-dimensional subspaces of $\Fq^{\ambdim}$. By properties of the Gaussian binomial coefficient, we get
$\textstyle\qbinomial{\ambdim-t_i}{s_i-t_i}\qbinomial{\ambdim}{s_i}^{-1} = \qbinomial{s_i}{t_i} \qbinomial{\ambdim}{t_i}^{-1}$.
The bounds immediately follow from \eqref{eq:bounds_gaussian_binomial}.
\end{proof}

\cref{lem:success_fixed_s_vector} allows us to compute the worst-case number of iterations of the algorithm for a given probability mass function $\ps$ of $\s$ by
\begin{align}
\NumItWorstCase = \max_{\t \in \Tsett} \left(\sum_{\s \in \Tsets} \ps \rhoprob(\s,\t)\right)^{-1}. \label{eq:worst_case_num_iterations_exact_optimization}
\end{align}

The problem of minimizing \eqref{eq:worst_case_num_iterations_exact_optimization} over all valid distributions $\ps$ on $\Tsets$ can be formulated as a linear program and solved numerically for small parameters $\ell$, $\ambdim$, $s$ using standard methods. Note that the unknowns are the $\ps \in [0,1]$, and the number of unknowns, $|\Tsets|$, grows fast in $\ell$, $\ambdim$, and $s$. Due to this limitation, we present a formal discussion in Appendix~\ref{app:optimal_support_finding_via_linear_programming} of this ``optimal'' choice of $\ps$, and continue with a more scalable solution.

We relax the problem of maximizing \eqref{eq:worst_case_num_iterations_exact_optimization} as follows.
\begin{itemize}
\item We give a randomized mapping $\scomp \, : \, \Tsett \to \Tsets$ that maximizes $\rhoprob(\scomp(\t,s),\t)$ for a given $\t \in \Tsett$ (see Algorithm~\ref{lem:success_fixed_s_vector} and \cref{lem:scomp_maximizes_rho} below).
This mapping is randomized, i.e.~for each input there are multiple possible outputs and one is selected at random; we discuss this further below.
\item Instead of choosing a vector $\s \in \Tsets$ directly, we first choose a vector $\t \in \Tsets$ at random according to a designed distribution $p_{\t}$ on $\Tsett$, and set $\s \gets \scomp(\t,s)$. This means that for a fixed error $\e$, we can bound
\begin{align*}
\Pr(\Esete \subseteq \Fset) &= \sum_{\s \in \Tsets} \ps \rhoprob(\s,\te)  \\
&\geq p_{\te} \cdot \rhoprob(\scomp(\te,s),\te).
\end{align*}
This bound is relatively tight for this choice of $\s$ (see \cref{prop:bounds_success_probability} below). %
\item Instead of minimizing  \eqref{eq:worst_case_num_iterations_exact_optimization}, we minimize the following upper bound on the worst-case expected number of iterations
\begin{align}
&\NumItWorstCase \leq \max_{\t \in \Tsett} \left[p_{\t} \cdot \rhoprob(\scomp(\t,s),\t)\right]^{-1}, \label{eq:worst_case_num_iterations_bound_optimization}
\end{align}
over all valid probability mass functions $p_\t$ on $\Tsett$.
\end{itemize}
This comes at the cost of a slightly smaller success probability than the optimal choice of $\ps$ (cf.~\cref{sec:comparison} for a numerical comparison), but allows us to give a support drawing strategy that can be practically implemented and whose running time we can bound.

\cref{alg:scomp} formally defines the randomized mapping $\scomp$ and \cref{lem:scomp_maximizes_rho} proves that $\s = \scomp(\t,s)$ maximizes $\rhoprob(\s,\t)$ among all $\s \in \Tsets$.
The randomization in \cref{line:scomp_choice_position_end} prevents a bias in preferring certain positions (compared to some deterministic choice), and seems to be practically advantageous, especially for large $\ell$: in fact, for the Hamming case with $\mu = 1$ and $n = \ell$, then such a randomization is essential for the efficacy of Prange's generic decoder (cf.~\cref{ssec:comparison_Hamming_rank}).
Our analysis, however, is not able to take the randomness properly into account, and will depend merely on $\rhoprobt(\t)$, which is defined as
\begin{align}
\rhoprobt(\t) := \rhoprob(\scomp(\t,s),\t) \label{eq:definition_rho_scomp}
\end{align}
for all $\t \in \Tsett$ and a fixed $s \geq t$.
Note that though $\scomp$ is randomized, then $\rhoprobt(\t)$ is not.

\begin{algorithm}[ht!]
	\caption{$\scomp(\t,s)$}\label{alg:scomp}
	\SetKwInOut{Input}{Input}\SetKwInOut{Output}{Output}
	\Input{$\t \in \Tsett$ and $s \in \ZZ$ with $t \leq s \leq \ell \mu$.}
	\Output{$\s \in \Tsets$}
	$\s = [s_1,\dots,s_\ell] \gets \t$; \quad $\delta \gets s-t$ \label{line:start_s_comp} \\
	\While{$\delta >0$}{
		\rev{$\Jset_1 \gets \big\{i \in\{1,\dots,n\} \, : \, s_i \neq \ambdim\big\}$ \label{line:scomp_choice_position_start}\\
			$\Jset_2 \gets \big\{i \in\Jset_1 \, : \, t_i = \displaystyle\max_{j  \in \Jset_1}\big\{t_j\big\} \big\}$ \\
			$\Jset_3 \gets \big\{i \in\Jset_2 \, : \, s_i = \displaystyle\min_{j  \in \Jset_2}\big\{s_j\big\} \big\}$ \\
			$h \sample \Jset_3$ \label{line:scomp_choice_position_end}\\
			$s_h \gets s_h + 1$; \, $\delta \gets \delta-1$ \label{line:stop_s_comp}
		}
    }
    \Return{$\s$}
\end{algorithm}

\begin{lemma}\label{lem:scomp_maximizes_rho}
Let $\t \in \Tsett$ and let $t \leq s \leq \ell \mu$.
Then, $\s = \scomp(\t,s)$, with $\scomp$ as in Algorithm~\ref{alg:scomp}, maximizes $\rhoprob(\s,\t)$, i.e.,
\begin{align*}
\rhoprob(\scomp(\t,s),\t) = \max_{\s \in \Tsets} \rhoprob(\s,\t).
\end{align*}
\end{lemma}

\begin{proof}
As the denominator of \eqref{eq:definition_varrho} is independent of $\s$, it suffices to show that $\s=\scomp(\t,s)$ maximizes
\begin{align}
\prod_{i=1}^{\ell} \qbinomial{s_i}{t_i}. \label{eq:varrho_max_numerator}
\end{align}
for a given $\t$.
Say that we start with $\s = \t$ and increase entries of $\s$ by one until we have $\sum_{i=1}^{\ell} s_i = s$ (note that we can assume this since \eqref{eq:varrho_max_numerator} is zero if $s_i < t_i$ for some $i$).
We observe that \eqref{eq:varrho_max_numerator} is increased by a factor
\begin{align*}
\frac{\qbinomial{s_i+1}{t_i}}{\qbinomial{s_i}{t_i}}
\end{align*}
if we increase position $i$ of $\s$.
For $s_i \geq t_i$, we have
\begin{align*}
\frac{\qbinomial{s_i+1}{t_i}}{\qbinomial{s_i}{t_i}} &= \prod_{\mu=1}^{t_i} \frac{\left(\frac{q^{s_i+2-\mu}-1}{q^{\mu}-1}\right)}{\left(\frac{q^{s_i+1-\mu}-1}{q^{\mu}-1}\right)} \\
&= \frac{q^{s_i+1}-1}{q^{s_i-t_i+1}-1}
\end{align*}
For a fixed $t_i$, the quantity $\tfrac{q^{s_i+1}-1}{q^{s_i-t_i+1}-1}$ is monotonically decreasing in $s_i$, and we have
\begin{equation*}
q^{\rev{t_i}} < \frac{q^{s_i+1}-1}{q^{s_i-t_i+1}-1} < q^{\rev{t_i+1}}. %
\end{equation*}
It follows that the largest increase of \eqref{eq:varrho_max_numerator} is achieved by increasing a position $i$ with \rev{smallest} $s_i$ among those positions with largest $t_i$. Increasing such a position in a greedy fashion attains a global maximum since this choice will also maximize the possible increase in the following steps.
Hence, \eqref{eq:varrho_max_numerator} is maximized by iteratively increasing $s_i$ \rev{by one} such that $s_i \leq \ambdim$ and $\sum_{i=1}^{\ell} s_i \leq s$ for some $i$ with \rev{smallest} $s_i < \ambdim$ among those positions that have a maximal $t_i$.
This is exactly what $\scomp(\cdot, \cdot)$ does.
\end{proof}

\subsection{The Support-Drawing Algorithm}\label{ssec:support-finding_algorithm}

Based on the ideas presented above, \cref{alg:draw_random_support} outlines the support-drawing algorithm that we propose.
The probability distribution $p_\t$ is chosen to minimize the bound on the worst-case expected number of iterations in \eqref{eq:worst_case_num_iterations_bound_optimization}.
The following proposition presents bounds on the expected number of iterations.
Note that the lower and upper bound are independent of the error and differ by only a factor $|\Tsett|$, which is relatively small compared to the absolute values of the bounds for not too large $\ell$.
For notational convenience, we define the following value:
\begin{align}
\Q := \sum_{\t \in \Tset} \rhoprobt(\t)^{-1}. \label{eq:Q}
\end{align}

\begin{algorithm}[ht!]
	\caption{$\mathsf{DrawRandomSupport}(s,t,\ambdim)$}\label{alg:draw_random_support}
	\SetKwInOut{Input}{Input}\SetKwInOut{Output}{Output}
	\Input{Integers $t,s,\ambdim$ with $\mu \leq \ambdim$ and $t \leq s \leq \ell \mu$}
	\Output{$\Fset$ of weight $s$}
	Draw $\t \in \Tsett$ according to the distribution
	\begin{align*}
	p_{\t} := \rhoprobt(\t)^{-1} \Q^{-1} \quad \forall \, \t \in \Tsett, \qquad \textrm{ where $Q$ is defined as in \eqref{eq:Q}} %
	\end{align*} \label{line:draw_t_vector} \\
	$\s \gets \scomp(\t,s)$ \\
	$\Fset \sample \SupportSet(\s)$ \\ %
	\Return{$\Fset$}
\end{algorithm}

\begin{proposition}\label{prop:bounds_success_probability}
Let $\e$ be an error of sum-rank weight $t$ and let $s$ be an integer with $t \leq s \leq \ell \mu$.
If $\Fset$ is a super-support that is drawn by Algorithm~\ref{alg:draw_random_support} with input $t$ and $s$, then we have
\begin{align*}
|\Tset|^{-1} \Q
\leq \frac{1}{\Pr(\Esete \subseteq \Fset)}
\leq \Q \ ,
\end{align*}
where $\Q$ is defined as in \eqref{eq:Q}. %
\end{proposition}

\begin{proof}
Denote by $\ps$ the distribution of $\s = \scomp(\t,s)$, where $\t$ is a random variable with probability mass function $p_\t$.
By \eqref{eq:worst_case_num_iterations_exact_optimization}, we have
\begin{align*}
\Pr(\Esete \subseteq \Fset) %
&= \sum_{\t \in \Tsets} p_{\t} \rhoprob(\scomp(\t,s),\te) \\
&\geq p_{\te} \rhoprob(\scomp(\te,s),\te) \\
&= \Q^{-1}.
\end{align*}
This proves the upper bound on $\Pr(\Esete \subseteq \Fset)^{-1}$.
For the lower bound, we first observe that for all $\t \in \Tsett$, \cref{lem:scomp_maximizes_rho} implies
\begin{align*}
\rhoprob(\scomp(\t,s),\te) \leq \rhoprob(\scomp(\t,s),\t) = \rhoprobt(\t).
\end{align*}
This yields
\begin{align*}
\Pr(\Esete \subseteq \Fset) &= \sum_{\t \in \Tsets} p_{\t} \rhoprob(\scomp(\t,s),\te) \\
&\leq \sum_{\t \in \Tsets} p_{\t} \rhoprobt(\t) \\
&= \sum_{\t \in \Tsets} \Q^{-1} = |\Tset| \Q^{-1} \ ,
\end{align*}
which proves the claim.
\end{proof}

At first glance, the lower and upper bounds in \cref{prop:bounds_success_probability} appear infeasible to compute since the number of summands, $|\Tsett|$, may grow super-polynomially in $t$ (depending on $\ell$ and $\mu$).
Furthermore, it is at this point unclear how to efficiently implement Line~\ref{line:draw_t_vector} of Algorithm~\ref{alg:draw_random_support}.
Below, we answer these two questions, and also give a simple upper bound on $\Q$.

\subsection{A Simple Bound on the Success Probability}

We start with a simple bound on $\Q$ from \eqref{eq:Q}.

\begin{proposition}\label{prop:simplified_bound_lemma}
For any $t \leq s \leq \ell \mu$, we have
\begin{align*}
\max_{\t \in \Tset} \rhoprobt(\t)^{-1} \leq \gamma_q^\ell q^{t(\ambdim-\frac{s}{\ell})}.
\end{align*}
In particular,
\begin{align*}
\Q \leq \tbinom{\ell+t-1}{\ell-1}\gamma_q^\ell q^{t(\ambdim-\frac{s}{\ell})} \ ,
\end{align*}
where $\gamma_q \leq 3.5$ is defined as in \eqref{eq:gamma_q}.
\end{proposition}

\begin{proof}
By \eqref{eq:bounds_rho_s_t} in \rev{\cref{lem:success_fixed_s_vector}}, we have
\begin{align*}
\max_{\t \in \Tset} \rhoprobt(\t)^{-1} &\leq \gamma_q^{\ell} \max_{\t \in \Tset} \left\{ q^{\sum_{i=1}^{\ell} t_i(\ambdim-s_i)}  \mid \s = \scomp(\t,s) \right\} \\
&= \gamma_q^{\ell} q^{t \ambdim} q^{-\min_{\t \in \Tset} \left\{ \sum_{i=1}^{\ell} t_is_i \mid \s = \scomp(\t,s) \right\}}
\end{align*}
We claim that the last exponent satisfies:
\begin{align*}
\min_{\t \in \Tset} \left\{ \sum_{i=1}^{\ell} t_is_i \mid \s = \scomp(\t,s) \right\} \geq \frac{ts}{\ell}
\end{align*}
We will prove this by relaxing the variables to reals, and consider only the ordered vectors $\t$,
so define the set:
\begin{align*}
\Tset^{(\mathbb{R},\mathsf{ord})} := \left\{ \t \in \mathbb{R}_{\geq 0}^\ell \, : \, \sum_{i=1}^{\ell} t_i = t, \, t_i \leq \mu, \, t_1 \geq t_2 \geq \dots \geq t_\ell \right\}
\end{align*}
and the mapping
\rev{
\begin{align*}
\scomp^{(\mathbb{R})} \, : \, \Tset^{(\mathbb{R},\mathsf{ord})} &\to \mathbb{R}_{\geq 0}^\ell, \\
\t &\mapsto \big[
\underbrace{\ambdim,\dots,\ambdim}_{\idxh \text{ \rev{times}}}, 
\underbrace{t_{\idxh+1}+\xi+1,\hdots,t_{\idxh+\idxg}+\xi+1}_{\idxg \text{ \rev{times}}},
\underbrace{t_{\idxh+\idxg+1}+\xi+\delta,\hdots,t_{\idxh+\idxf}+\xi+\delta}_{\idxf - \idxg \text{ \rev{times}}}
,t_{\idxh+\idxf+1},\dots,t_\ell \big],
\end{align*}
where
\begin{align*}
\idxh &:= \max \left\{ h' \in\{0,1,\hdots,\lsr\} \, : \, \sum_{i=1}^{h'}(\ambdim-t_i) \leq s-t, t_{h'}>t_{h'+1} \text{ with } t_{0}:=\ambdim,t_{\lsr+1}:=-1 \right\}, \\
\idxf &:= \max\{f'\in\{1,\hdots,\lsr\}\,:\, t_{f'}=t_{\idxh+1} \}-\idxh, \\
\srem &:= s-t-\sum_{i=1}^{h} (\ambdim-t_i), \\ %
\xi & := \left \lfloor \frac{\srem}{f} \right\rfloor, \\
\idxg & := \lfloor \srem \rfloor - \xi f, \\ 
\delta & := \frac{\srem-\lfloor \srem \rfloor}{\idxf - \idxg}. 
\end{align*}}
Note that $\scomp^{(\mathbb{R})}$ agrees with a deterministic variant of\footnote{The outputs are equal if we choose $j \gets \min \big\{j \, : \, s_j = \displaystyle\max_{i \, : \, s_i \neq \ambdim}\{s_i\}\big\}$ instead of a random choice in \cref{line:scomp_choice_position_end} of \cref{alg:scomp}. Note that in what follows here, the choice of $j$ is irrelevant, so we may w.l.o.g.\ assume that $j$ is chosen like this.} $\scomp$ on $\Tset^{(\mathbb{R},\mathsf{ord})} \cap \ZZ^\ell$.
Since $\sum_{i=1}^{\ell} t_is_i|_{\s = \scomp(\t,s)}$ is independent of the ordering of the entries of $\t$ and the set of sorted elements (vectors) of $\Tsett$ are subset of $\Tset^{(\mathbb{R},\mathsf{ord})}$, we have
\begin{align*}
\min_{\t \in \Tset} &\left\{ \sum_{i=1}^{\ell} t_is_i \mid \s = \scomp(\t,s) \right\} \geq \min_{\t \in \Tset^{(\mathbb{R},\mathsf{ord})}} \left\{ \sum_{i=1}^{\ell} t_is_i \mid \s = \scomp^{(\mathbb{R})}(\t,s) \right\}.
\end{align*}
For $\t \in \Tset^{(\mathbb{R},\mathsf{ord})}$ and $\s = \scomp^{(\mathbb{R})}(\t,s)$, we have 
\rev{
\begin{align}
\sum_{i=1}^{\ell} t_is_i 
&= \ambdim \sum_{i=1}^{\idxh} t_i + \sum_{i=\idxh+1}^{\idxh+\idxg}(t_{i} +\xi+1)  t_{i} + \sum_{i=\idxh+\idxg+1}^{\idxh+\idxf}(t_{i} +\xi+\delta)  t_{i} + \sum_{i=\idxh+\idxf+1}^{\ell} t_i^2  \label{eq:sum_tisi_real_firsteq}\\
&= \ambdim \sum_{i=1}^{\idxh} t_i + \idxg(t_{\idxh+1} +\xi+1) t_{\idxh+1} + (f-g)(t_{\idxh+1} +\xi+\delta)  t_{\idxh+1} + \sum_{i=\idxh+\idxf+1}^{\ell} t_i^2 \notag\\
&= \ambdim \sum_{i=1}^{\idxh} t_i + \idxg(\xi+1) t_{\idxh+1} + (\idxf-\idxg)(\xi+\delta)  t_{\idxh+1} + \sum_{i=\idxh+1}^{\ell} t_i^2 \notag \\
&= \ambdim \sum_{i=1}^{\idxh} t_i + (\underbrace{\xi \idxf+g}_{=\lfloor \srem \rfloor}  + \underbrace{\delta (\idxf-\idxg)}_{=\srem-\lfloor \srem \rfloor} )t_{\idxh+1} + \sum_{i=\idxh+1}^{\ell} t_i^2 \notag \\
&= \ambdim \sum_{i=1}^{\idxh} t_i + \srem t_{\idxh+1} + \sum_{i=\idxh+1}^{\ell} t_i^2. \label{eq:sum_tisi_real}
\end{align}
Since $ t_i \leq t_{i+1} +\xi+\delta\leq t_{i+1} +\xi+1 \leq \ambdim$, it follows that \eqref{eq:sum_tisi_real_firsteq} is minimized by a sequence in $\Tset^{(\mathbb{R},\mathsf{ord})}$ with smallest-possible $\idxh$.
Among these sequences with minimal $\idxh$, it is minimized by sequence with largest $\idxf$.
Since $t_i$ are non-increasing, these requirements directly imply that \eqref{eq:sum_tisi_real} is minimized for
\begin{align*}
\t = \big[ \tfrac{t}{\ell}, \dots, \tfrac{t}{\ell} \big],
\end{align*}
for which we have
\begin{align*}
\sum_{i=1}^{\ell} t_is_i = \frac{t}{\ell} \sum_{i=1}^{\ell} s_i = \frac{ts}{\ell}.
\end{align*}}
This proves the first claim.

We get the bound on $\Q$ by
\begin{align*}
\Q = \sum_{\t \in \Tset} \rhoprobt(\t)^{-1} &\leq |\Tsett| \max_{\t \in \Tset} \rhoprobt(\t)^{-1} \leq \tbinom{\ell+t-1}{\ell-1}\gamma_q^\ell q^{t(\ambdim-\frac{s}{\ell})} \qedhere %
\end{align*}
\end{proof}

\subsection{Computing Bounds on the Success Probability Efficiently}

We turn to the question of computing $\Q$, as in \eqref{eq:Q},
exactly.
Below we give a dynamic-programming algorithm that computes this sum efficiently using a recursion formula.
The algorithm is similar to the counting algorithm for the number of vectors of a given sum-rank weight (Algorithm~\ref{alg:number_sum-rank-words}), with a major complication: we have
\begin{align*}
\rhoprobt(\t)^{-1} = \prod_{i=1}^{\ell} \frac{\qbinomial{\ambdim}{t_i}}{\qbinomial{\scomp(\t,s)_i}{t_i}},
\end{align*}
where $\scomp(\t,s)_i$ (the $i$-th entry of the vector $\scomp(\t,s)$) depends on the entire vector $\t$ and not only on $t_i$.
Hence, we cannot easily split the product into a part depending only on $t_1$ and one depending only on $t_2,\dots,t_\ell$.
Note, however, that if $\t$ is ordered in non-decreasing order and $j$ is such that $t_j>t_{j+1}$, then 
\begin{align*}
\prod_{i=1}^{j} \qbinomial{\scomp(\t,s)_i}{t_i}
\end{align*}
depends only on $t_1,\dots,t_j$, and is invariant under the randomness of $\scomp$.
This motivates the following statement, for which we define the following two notions:
\begin{align*}
\Tsettord &:= \left\{ \t \in \Tsett \, : \, t_1 \geq t_2 \geq \dots \geq t_\ell \right\}, \\
\delta_i(\t) &:= \left| \left\{ j \, : \, t_j = i \right\}\right| \quad \forall \, i=0,\dots,\mu, \, \t \in \Tsett.
\end{align*}

\begin{lemma}\label{lem:recursion_sum_rhos}
For all $\ell \geq 1$, $s \leq \ell \ambdim$, and $0 \leq t \leq \min\{s,\ell \mu\}$, we have
\begin{align}
\Q =
\ell! \cdot \Msum(t,\ell,\nmmin,s), \label{eq:sum_rho}
\end{align}
where for any $t', \ell',\mu',s' \in \ZZ_{\geq 0}$, we define
\begin{align*}
&\Msum(t',\ell',\nmmin',s') :=
\begin{cases}
\displaystyle\sum_{\t \in \TsetordWITHARGUMENTS{t',\ell',\mu'}} \left(\prod_{i=0}^{\nmmin'} \delta_i(\t)! \right)^{-1} \prod_{i=1}^{\ell'}\frac{\qbinomial{\ambdim}{t_i}}{\qbinomial{\scomp(\t,s')_i}{t_i}}, &\substack{\ell'\geq 1 \\ 0 \leq t' \leq \min\{s', \,\ell' \mu'\} \\ s' \leq \ell' \ambdim}, \\
1, &\ell'=t'=s'=0, \\
0, &\text{else}.
\end{cases}
\end{align*}
Furthermore, $\Msum(t',\ell',\nmmin',s')$ fulfills the following recursive relation.
\begin{align*}
\Msum(t',\ell',\nmmin',s') = \sum_{t_1 = \lceil \tfrac{t'}{\ell'} \rceil}^{\min\{\mu',t'\}}
\sum_{\delta = \max\{t'-\ell' (t_1-1), 1\} }^{\max\{\delta \, : \, \delta \leq \ell', \, t_1\delta \leq t'\}}
&\left( \tfrac{1}{\delta!}\prod_{i=1}^{\delta} \frac{\qbinomial{\ambdim}{t_1}}{\qbinomial{\scomp([\overbrace{t_1,\dots,t_1}^{\text{$\delta$ times}}],\min\{s'-(t'-\delta t_1),\delta\ambdim\})_i}{t_1}}\right) \\
&\quad \cdot \Msum\Big(t'-\delta t_1,\ell'-\delta,t_1-1,s'-\min\{s'-(t'-\delta t_1),\delta\ambdim\}\Big). %
\end{align*}
\end{lemma}

\begin{proof}
Equation \eqref{eq:sum_rho} holds since, by definition, we have
\begin{align*}
\Q = \sum_{\t \in \Tsett} \rhoprobt(\t)^{-1} = \sum_{\t \in \Tsett}\prod_{i=1}^{\ell}\frac{\qbinomial{\ambdim}{t_i}}{\qbinomial{\scomp(\t,s)_i}{t_i}}.
\end{align*}
Furthermore, the term $\prod_{i=1}^{\ell}\qbinomial{\ambdim}{t_i} \qbinomial{\scomp(\t,s)_i}{t_i}^{-1}$ is invariant under permutations of $\t$, so we can group these summands into those that belong to a unique sorted vector $\t \in \Tsettord$. The number of these summands belonging to the same sorted $\t$ equals the number of permutations of $\t$, which is $\tfrac{\ell!}{\prod_{i=0}^{\nmmin} \delta_i(\t)!}$. This proves \eqref{eq:sum_rho}.

The recursion formula is correct by the following argument.
The restrictions on the choice of $t_1$ and $\delta$ are as follows (which directly yield the limits of the sums):
\begin{itemize}
\item $t_1 \leq \min\{t',\mu'\}$ by definition of $\TsetordWITHARGUMENTS{t',\ell',\mu}$.
\item $t_1 \geq \tfrac{t'}{\ell'}$ since for a given $t_1$, the entire vector $\t$ may only sum up to at most $t_1\ell'$ (since $\delta\leq\ell'$ and $t_i \leq t_1$). On the other hand, the entries of the vector must sum to $t'$, which is impossible for $t_1 \ell' < t'$.
\item $1 \leq \delta \leq \ell'$ since $t_1$ may appear between $1$ and $\ell'$ times.
\item $t_1\delta \leq t'$ since $\t$ sums to $t'$ and thus we must have $t_1\delta \leq t'$.
\item $\delta \geq t'-\ell'(t_1-1)$ since the remaining entries of $\t$ have values $\leq t_1-1$ and must nevertheless sum to $t'$. This is only possible for $(\ell'-\delta)(t_1-1) \geq t'-t_1' \delta$, which is equivalent to $\delta \geq t'-\ell'(t_1-1)$.
\end{itemize}

For fixed $t_1$ and $\delta<\ell'$, a vector $\t \in \TsetordWITHARGUMENTS{t',\ell',\mu}$ whose first $\delta$ positions equal $t_1$ and whose remaining positions are $\leq t_1-1$ can be split into two parts $\t = \big[\t^{(1)} \mid \t^{(2)} \big]$, where
\begin{align*}
\t^{(1)} &:= [t_1,\dots,t_1] \in \ZZ^{\delta}, \\
\t^{(2)} &:= [t_{\delta+1},\dots,t_{\ell'}] \in \{0,\dots,t_1-1\}^{\ell'-\delta}.
\end{align*}
In particular, we have
\begin{align*}
\t^{(2)} \in \TsetordWITHARGUMENTS{t'-t_1\delta, \ell'-\delta, t_1-1}.
\end{align*}
Hence, we can split up the product
\begin{align*}
\prod_{i=0}^{\nmmin'} \delta_i(\t)! = \delta! \cdot \prod_{i=0}^{t_1-1} \delta_i\big(\t^{(2)}\big)!.
\end{align*}
Furthermore, note that by definition of $\scomp$, we have that $\big[\s^{(1)} \mid \s^{(2)} \big]$ is a valid output of $\scomp(\t, s')$,
where
\begin{align*}
\s^{(1)} &= \scomp\big(\t^{(1)}, \min\{s'-(t'-\delta t_1),\delta\ambdim\}\big), \\ %
\s^{(2)} &= \scomp\big(\t^{(2)}, s'-\min\{s'-(t'-\delta t_1),\delta\ambdim\}\big). %
\end{align*}
In particular, $\s^{(1)}$ and $\s^{(2)}$ only depend on $\t^{(1)}$ and $\t^{(2)}$, respectively, and on the parameters $s'$, $t_1$, and $\delta$.
Hence, we can also split the product
\begin{align*}
&\prod_{i=1}^{\ell'}\frac{\qbinomial{\ambdim}{t_i}}{\qbinomial{\scomp(\t,s')_i}{t_i}} = \left(\prod_{i=1}^{\delta}\frac{\qbinomial{\ambdim}{t_1}}{\qbinomial{\scomp(\t^{(1)}, \min\{s'-(t'-\delta t_1),\delta\ambdim\})_i}{t_1}}\right)
\cdot \left( \prod_{i=1}^{\ell'-\delta}\frac{\qbinomial{\ambdim}{t_i^{(2)}}}{\qbinomial{\scomp(\t^{(2)}, s'-\min\{s'-(t'-\delta t_1),\delta\ambdim\})_i}{t_i^{(2)}}}\right)
\end{align*}

For $\delta=\ell'$, we have $\t = [t_1,\dots,t_1]$. Hence, we get
\begin{align*}
&\left(\prod_{i=0}^{\nmmin} \delta_i(\t)! \right)^{-1} \prod_{i=1}^{\ell'}\frac{\qbinomial{\ambdim}{t_i}}{\qbinomial{\scomp(\t,s')_i}{t_i}}
= \begin{cases}
\tfrac{1}{\delta!}\prod_{i=1}^{\delta} \frac{\qbinomial{\ambdim}{t_1}}{\qbinomial{\scomp(\t,s')_i}{t_1}}, &\text{if } t' = \delta t_1 \text{ and } s' \leq \delta \ambdim, \\
0, &\text{else}.
\end{cases}
\end{align*}
By definition of the base case,
\begin{align*}
\Msum\big(t',0,\mu',s'\big) :=
\begin{cases}
1, &\text{if } t'=0 \text{ and } s'=0, \\
0, &\text{else},
\end{cases}
\end{align*}
we get exactly this summand for $\delta=\ell'$.
This proves the recursion.
\end{proof}

\begin{algorithm*}
\caption{Fill table $\{\Msum(t',\ell',\nmmin',s')\}_{t'\leq t, \ell' \leq \ell}^{\nmmin' \leq \nmmin, s' \leq s}$}\label{alg:fill_M_table}
\SetKwInOut{Input}{Input}\SetKwInOut{Output}{Output}
\Input{Integers $t'\leq t,\ell' \leq \ell,\nmmin' \leq \mu,s' \leq s$, global table $\{\Msum(t',\ell',\nmmin',s')\}_{t'\leq t, \ell' \leq \ell}^{\nmmin' \leq \nmmin, s' \leq s}$, global parameters $q$, $\ambdim$}
\Output{$\Msum(t',\ell',\nmmin',s')$}
\If{$\Msum(t',\ell',\nmmin',s')=-1$}{
	\If{$\ell'=t'=s'=0$ \label{line:fill_M_table_start_computation}}{
		$res \gets 1$
	} \Else {
		\If{$\ell' \geq 1$ and $0 \leq t' \leq \min\{s',\ell' \mu'\}$ and $s' \leq \ell' \ambdim$} {
			$res \gets 0$ \\
			\For{$t_1 = \lceil \tfrac{t'}{\ell'} \rceil, \dots, \min\{\mu',t'\}$} {
				\For{$\delta=\max\{t'-\ell' (t_1+1), 1\}, \dots, {\max\{\delta \, : \, \delta \leq \ell', \, t_1\delta \leq t'\}}$} {
					$s^{(1)} \gets \scomp([t_1,\dots,t_1], \min\{s'-(t'-\delta t_1),\delta\ambdim\})$ \label{line:fill_M_table_scomp} \\
					$res \gets res + \Msum\Big(t'-\delta t_1,\ell'-\delta,t_1-1,s'-\min\{s'-(t'-\delta t_1),\delta\ambdim\}\Big) \cdot \delta!^{-1} \cdot \qbinomial{\ambdim}{t_1}^{\delta} \cdot \prod_{i=1}^{\delta} \qbinomial{s^{(1)}_i}{t_1}^{-1}$  \label{line:fill_M_table_recursion}
				}
			}
		} \Else {
			$res \gets 0$
		}
	}
	$\Msum(t',\ell',\nmmin',s') \gets res$  \label{line:fill_M_table_end_computation}\\
}
\Return $\Msum(t',\ell',\nmmin',s')$
\end{algorithm*}

\begin{proposition}\label{prop:sum_compute}
If we initialize a table $\{\Msum(t',\ell',\nmmin',s')\}_{t'\leq t, \ell' \leq \ell}^{\nmmin' \leq \nmmin, s' \leq s}$ with $\Msum(t',\ell',\nmmin',s')=-1$ for all entries and call Algorithm~\ref{alg:fill_M_table} with input $t,\ell,\nmmin,s$, then the algorithm computes the entry $\Msum(t,\ell,\nmmin,s)$ in %
\begin{align*}
O^\sim\!\left(tsn^3\mu \ambdim^2 \log_2(q)\right),
\end{align*}
bit operations.
In particular, we can compute $\Q$ from \eqref{eq:Q}
in this bit complexity.
\end{proposition}

\begin{proof}
The correctness of the algorithm follows from Lemma~\ref{lem:recursion_sum_rhos}.
For the complexity, we observe the following:
Lines~\ref{line:fill_M_table_start_computation}--\ref{line:fill_M_table_end_computation} of the algorithm are only once called for each table index $[t',\ell',\mu',s']$. The number of table entries, and thus the calls of these expensive lines, is in $O(t s \mu \ell) \subseteq O(tsn)$.
It is negligible compared to the entire recursive call of the algorithm to pre-compute the products $\delta!^{-1} \cdot \qbinomial{\ambdim}{t_1}^\delta \cdot \prod_{i=1}^{\delta} \qbinomial{s^{(1)}([t_1,\dots,t_1], s')_i}{t_1}^{-1}$ for all $0 \leq \delta\leq \ell$, $0 \leq t_1 \leq \min\{t,\nmmin\}$, and $t_1\delta \leq s' \leq \ambdim \delta$.

The bottleneck of the algorithm is Line~\ref{line:fill_M_table_recursion}, where we multiply two rational numbers and add the result to another rational number. All these rational numbers are in $\iota^{-1} \ZZ$, where
\begin{align*}
\iota &= \ell! \left(\prod_{t'=0}^{\nmmin} \prod_{s'=t'}^{\ambdim} \qbinomial{s'}{t'}\right)^\ell
\leq \ell! 4^{\ell \nmmin \ambdim} q^{\sum_{t'=0}^{\nmmin} \sum_{s'=t'}^{\ambdim} t'(s'-t')}
\leq 2^{\ell \log_2(\ell) + 2 \ell \nmmin \ambdim + \ell \nmmin^2 \ambdim^2\log_2(q)}.
\end{align*}
Hence, we can implement all operations in $\iota^{-1} \ZZ$, and operations have a quasi-linear cost \cite{harvey2019faster} in the size of the numerators plus the size of $\iota$.
Furthermore, the numerator is \rev{upper} bounded by $\iota \Q$, %
which is again upper bounded by the bound in \cref{prop:simplified_bound_lemma}.
Thus, Line~\ref{line:fill_M_table_recursion} costs
\begin{align*}
&O^\sim\!\left( t(\ambdim-\tfrac{s}{\ell}) \log_2(q)+(\ell-1)\log_2(t+\ell-1)  + \ell \mu^2 \ambdim^2 \log_2(q)\right)
\subseteq O^\sim(n \mu \ambdim^2 \log_2(q))
\end{align*}
bit operations. %

Since Line~\ref{line:fill_M_table_recursion} is called $O(\ell \mu) \subseteq O(n)$ times for each table entry, the overall bit complexity of the entire recursion is
\begin{align*}
O^\sim\!\left(tsn^3\mu \ambdim^2 \log_2(q)\right),
\end{align*}
which proves the claim.
\end{proof}

\subsection{Efficiently Drawing Decomposition Vectors}\label{ssec:efficient_drawing}

With the help of Algorithm~\ref{alg:fill_M_table} and a bit of extra work, we can draw efficiently from the distribution $p_{\t}$ as in Algorithm~\ref{alg:draw_random_support}.
The idea of the method (see Algorithm~\ref{alg:draw_efficiently_from_tprime_distribution} below) is based on enumerative encoding \cite{cover1973enumerative}.
To formalize the idea, we need the following notation.
We denote by $\t \leq \t'$ for $\t,\t' \in \ZZ^{\ell}$ the lexicographical (total) ordering on $\ZZ^{\ell}$.
For $\t \in \ZZ^{\ell}$, $\t' \in \ZZ^{\ell'}$, and $i \leq \min\{\ell,\ell'\}$, we define the preorder $\t \leq_i \t'$ as
\begin{align*}
[t_1,\dots,t_i] \leq [t_1',\dots,t_i'].
\end{align*}
Further, we write $\t =_i \t'$ if $\t \leq_i \t'$ and $\t' \leq_i \t$, as well as $\t <_i \t'$ if $\t \leq_i \t'$, but not $\t =_i \t'$.
The following lemma shows how to compute the sum
\begin{align*}
\sum_{\substack{\tilde{\t} \in \Tsettord \\ \tilde{\t} =_{\ell'} \t}} \tfrac{\ell!}{\prod_{i=0}^{\mu} \delta_i(\tilde{\t})!} \rhoprobt(\tilde{\t})^{-1}
\end{align*}
efficiently, where we only sum over those vectors $\tilde{\t} \in \Tsettord$ who have a given prefix $\t$ of length $\ell'$.
This is a key ingredient for the enumerative-coding-based drawing method presented below.

\begin{lemma}\label{lem:sum_with_prefix}
Let $1\leq \ell' \leq \ell$ and $\t \in \TsetWITHARGUMENTS{t,\ell',\nmmin}^{(\mathsf{ord})}$.
Denote by $t_{\ell'}$ the $\ell'$-th entry of $\t$ and by $1 \leq \delta \leq \ell'$ the number of times $t_{\ell'}$ occurs in $\t$.
Thus, we can split $\t$ into
\begin{align*}
\t := \left[ \t^{(1)}, \t^{(2)} \right],
\end{align*}
where $\t^{(1)} \in \{t_{\ell'}+1,\dots, \nmmin\}^{\ell'-\delta}$ and $\t^{(2)} = \left[t_{\ell'},\dots,t_{\ell'}\right] \in \ZZ^{\delta}$.
Write $t^{(1)} := \sum_{i=1}^{\ell'-\delta} t^{(1)}_i$ and $s^{(1)} := \min\{s-t+t^{(1)},(\ell'-\delta)\ambdim\}$.
Then,
\begin{align*}
\sum_{\substack{\tilde{\t} \in \Tsettord \\ \tilde{\t} =_{\ell'} \t}} \tfrac{\ell!}{\prod_{i=0}^{\mu} \delta_i(\tilde{\t})!} \rhoprobt(\tilde{\t})^{-1}
&= \tfrac{\ell!}{\prod_{i=t_{\ell'}+1}^{\mu} \delta_i\left(\t^{(1)}\right)!} \left(\prod_{j=1}^{\ell'-\delta}\frac{\qbinomial{\ambdim}{t_i^{(1)}}}{\qbinomial{\scomp(\t^{(1)}, s^{(1)})_i}{t_i^{(1)}}}\right) \\
&\cdot 
\sum_{\delta' = \max\{t-t^{(1)}-(t_{\ell'}-1)(\ell-\ell'+\delta), \delta \} }^{\max\{\delta' \, : \, \delta' \leq \ell-\ell'+\delta, \, t_{\ell'}\delta' \leq t-t^{(1)}\}}
\left( \tfrac{1}{\delta'!}\prod_{i=1}^{\delta'} \frac{\qbinomial{\ambdim}{t_{\ell'}}}{\qbinomial{\scomp([\overbrace{t_{\ell'},\dots,t_{\ell'}}^{\text{$\delta'$ times}}],s^{(2)}(\delta'))_i}{t_{\ell'}}}\right) \\
&\cdot \Msum\Big(t-\delta' t_{\ell'}-t^{(1)},\ell-(\ell'-\delta+\delta'),t_{\ell'}-1,s-s^{(1)}-s^{(2)}(\delta')\Big),
\end{align*}
where $s^{(2)}(\delta') := \min\{s-s^{(1)}-(t-\delta' t_{\ell'}-t^{(1)}), \delta'\ambdim\}$ and $\Msum(t',\ell',\nmmin',s')\}$ is defined as in Lemma~\ref{lem:recursion_sum_rhos}.

In particular, if the table $\{\Msum(t',\ell',\nmmin',s')\}_{t'\leq t, \ell' \leq \ell}^{\nmmin' \leq \nmmin, s' \leq s}$ is pre-computed, we can compute $\sum_{\substack{\tilde{\t} \in \Tsettord \\ \tilde{\t} =_{\ell'} \t}} \tfrac{\ell!}{\prod_{i=0}^{\mu} \delta_i(\tilde{\t})!} \rhoprobt(\tilde{\t})^{-1}$
in
\begin{align*}
O^\sim(n^2 \ambdim^2 \log_2(q))
\end{align*}
bit operations.
\end{lemma}

\begin{proof}
The statement follows by the same arguments as the recursive formula for $\Msum(\cdot,\cdot,\cdot,\cdot)$ in \cref{lem:recursion_sum_rhos}.
The only difference is that we split the sum (only) into those subsets of $\{\tilde{\t} \in \Tsettord \, : \, \tilde{\t} =_{\ell'} \t\}$ in which the value $t_{\ell'}$ occurs exactly the same number of times $\delta'$. %
Since we know that $t_{\ell'}$ is contained $\delta$ times in the last positions of the prefix vector, it must occur $\delta'\geq \delta$ times in $\tilde{\t}$.
Furthermore, $\delta'$ must be chosen large enough such that
\begin{equation*}
t-t^{(1)} \leq t_{\ell'}\delta'+(t_{\ell'}-1)(\ell-(\ell'-\delta+\delta')),
\end{equation*}
which gives the other lower bound (and sum limit) on $\delta'$.
On the other hand, we must have $\delta' \leq \ell-\ell'+\delta$ since the length of $\tilde{\t}$ is $\ell$ and the length of $\t^{(1)}$ is $\ell'-\delta$.
After subtracting the sum of the entries $>t_{\ell'}$ of the prefix vector, the remaining part of the vector $\tilde{\t}$ can only sum up to $t-t^{(1)}$. In particular, we must have $t_{\ell'} \delta' \leq t-t^{(1)}$. This gives the upper bound (and sum limit) on $\delta'$.

The formula follows by splitting the product (w.r.t.~$j$) in
\begin{align*}
\tfrac{\ell!}{\prod_{i=0}^{\mu} \delta_i(\tilde{\t})!} \rhoprobt(\tilde{\t})^{-1} = \tfrac{\ell!}{\prod_{i=0}^{\mu} \delta_i(\tilde{\t})!} \prod_{j=1}^{\ell} \frac{\qbinomial{\ambdim}{\tilde{t}_j}}{\qbinomial{\scomp(\tilde{\t},s)_j}{\tilde{t}_j}}
\end{align*}
into the following subsets of positions $j$:
\begin{itemize}
\item the positions of the prefix vector with values $t_j > t_{\ell'}$,
\item the part of $\tilde{\t}$ in which $t_{\ell'}$ is repeated $\delta'$ times, and
\item the remaining part of $\tilde{\t}$.
\end{itemize}
The sum over the latter part is given by $\Msum\Big(t-\delta' t_{\ell'}-t^{(1)},\ell-\ell'+\delta'-\delta,t_{\ell'}-1,s-s^{(1)}-s^{(2)}(\delta')\big\}\Big)$ since this part of $\tilde{\t}$ must sum up to $t-\delta' t_{\ell'}-t^{(1)}$, it is a vector of length $\ell-(\ell'-\delta+\delta')$, we have $\tilde{t}_i < t_{\ell'}$ for these entries of $\tilde{\t}$.
The choices of $s^{(1)}$, $s^{(2)}(\delta')$, and $s-s^{(1)}-s^{(2)}(\delta')$ are to ensure that
\begin{align*}
&\Big[\scomp(\t^{(1)}, s^{(1)}) \mid \scomp([\overbrace{t_{\ell'},\dots,t_{\ell'}}^{\text{$\delta'$ times}}],s^{(2)}(\delta')) \mid
\scomp([\tilde{t}_{\ell'-\delta+\delta'+1},\dots,\tilde{t}_{\ell}], s-s^{(1)}-s^{(2)}(\delta') \Big]
\end{align*}
is a valid output of $\scomp(\tilde{\t},s)$ (i.e., independent of $\scomp$'s randomness, we can split the product $\prod_{i=1}^{\ell}\qbinomial{\scomp(\tilde{\t},s)_i}{t_i}$ into the given three parts).

Complexity-wise, the bottleneck are at most $\ell$ multiplications and additions of rational numbers in $\iota^{-1} \ZZ$, where $\iota$ is the same as in the proof of \cref{prop:sum_compute}. Also the numerators of all involved rational numbers are bounded as in \cref{prop:sum_compute}.
Hence, computing $\sum_{\tilde{\t} \in \Tsettord, \tilde{\t} =_{\ell'} \t} \tfrac{\ell!}{\prod_{i=0}^{\mu} \delta_i(\tilde{\t})!} \rhoprobt(\tilde{\t})^{-1}$
costs $O^\sim(\ell n \mu \ambdim^2 \log_2(q)) \subseteq O^\sim(n^2 \ambdim^2 \log_2(q))$ bit operations.
\end{proof}

\begin{algorithm}
\caption{Draw Efficiently from Distribution $p_{\t}$ as in Algorithm~\ref{alg:draw_random_support}}\label{alg:draw_efficiently_from_tprime_distribution}
\SetKwInOut{Input}{Input}\SetKwInOut{Output}{Output}
\Input{Parameters $q,\ambdim,t,s,\ell,\mu$, precomputed table $\{\Msum(t',\ell',\nmmin',s')\}_{t'\leq t, \ell' \leq \ell}^{\nmmin' \leq \nmmin, s' \leq s}$}
\Output{$\t \in \Tsett$, drawn at random from the distribution ($\Q$ as in \eqref{eq:Q})
\begin{align*}
p_{\t} = \rhoprobt(\t)^{-1} \Q^{-1} \quad \forall \, \t' \in \Tsett.
\end{align*}}
$\iota \gets \ell! \left(\prod_{t'=0}^{\nmmin} \prod_{s'=t'}^{\ambdim} \qbinomial{s'}{t'}\right)^\ell$ \label{line:draw_efficiently_from_tprime_distribution_draw_uniform_x_from_interval} \\
$x \gets$ uniformly at random from the set of non-negative integers $< \iota \sum_{\tilde{\t} \in \Tset} \rhoprobt(\tilde{\t})^{-1}$  \\
$x \gets x/\iota$ \\
\For{$i=1,\dots,\ell$} {
	$t_i \gets \max\left\{ t'' \, : \, \displaystyle \sum_{t'=0}^{t''-1} \sum_{\substack{\tilde{\t} \in \Tsettord \\ \tilde{\t} =_{i} [t_1,\dots,t_{i-1},t']}} \tfrac{\ell!}{\prod_{i=0}^{\mu} \delta_i(\tilde{\t})!} \rhoprobt(\tilde{\t})^{-1} \leq x \right\}$ \\
	$x \gets x - \displaystyle \sum_{t'=0}^{t_{i-1}} \sum_{\substack{\tilde{\t} \in \Tsettord \\ \tilde{\t} =_{i} [t_1,\dots,t_{i}]}}  \tfrac{\ell!}{\prod_{i=0}^{\mu} \delta_i(\tilde{\t})!}  \rhoprobt(\tilde{\t})^{-1}$
}
$\t \gets \left[t_1,\dots,t_\ell\right]$ \label{line:draw_efficiently_from_tprime_distribution_sorted_vector_t_obtained} \\
$\pi \gets $ permutation drawn uniformly from the permutations of a multiset with set multiplicities $\delta_0(\t), \delta_1(\t), \dots, \delta_\mu(\t)$  \label{line:draw_efficiently_from_tprime_distribution_choose_random_permutation} \\
\Return $\pi(\t)$ \label{line:draw_efficiently_from_tprime_distribution_return_permuted_vector}
\end{algorithm}

\begin{proposition}\label{prop:drawing_efficiently_from_tprime_distribution}
Algorithm~\ref{alg:draw_efficiently_from_tprime_distribution} is correct and has complexity %
\begin{align*}
O^\sim(n^3 \ambdim^2 \log_2(q))
\end{align*}
bit operations. In particular, Line~\ref{line:draw_t_vector} of Algorithm~\ref{alg:draw_random_support} can be implemented with this complexity.
\end{proposition}

\begin{proof}
Since $p_{\t'} = p_{\t''}$ for two vectors $\t'$ and $\t''$ that are permutationally equivalent, we can simply draw a sorted vector from $\Tsettord$ using the probability mass function
\begin{align*}
\tilde{p}_{\t'} &:= \frac{\frac{\ell!}{\prod_{i=0}^{\mu} \delta_i(\tilde{\t})!} \rhoprob(\t')^{-1}}{\sum_{\tilde{\t} \in \Tset} \rhoprobt(\tilde{\t})^{-1}} = \frac{\frac{\ell!}{\prod_{i=0}^{\mu} \delta_i(\tilde{\t})!} \rhoprob(\t')^{-1}}{\sum_{\tilde{\t} \in \Tsettord} \frac{\ell!}{\prod_{i=0}^{\mu} \delta_i(\tilde{\t})!} \rhoprobt(\tilde{\t})^{-1}}
\end{align*}
for all $\t' \in \Tsettord$ (recall that $\tfrac{\ell!}{\prod_{i=0}^{\mu} \delta_i(\tilde{\t})!}$ is the number of permutations of the vector $\t'$). This is done in Lines~\ref{line:draw_efficiently_from_tprime_distribution_draw_uniform_x_from_interval}--\ref{line:draw_efficiently_from_tprime_distribution_sorted_vector_t_obtained}. Subsequently, we randomly permute this vector and obtain a vector that is drawn according to the distribution $p_{\t'}$ (see Lines~\ref{line:draw_efficiently_from_tprime_distribution_choose_random_permutation} and \ref{line:draw_efficiently_from_tprime_distribution_return_permuted_vector}).

The idea of Lines~\ref{line:draw_efficiently_from_tprime_distribution_draw_uniform_x_from_interval}--\ref{line:draw_efficiently_from_tprime_distribution_sorted_vector_t_obtained} is to partition the interval
\begin{align*}
\Iset := \left[0,\sum_{\tilde{\t} \in \Tsett} \rhoprobt(\tilde{\t})^{-1}\right) = \left[0,\sum_{\tilde{\t} \in \Tsettord} \tfrac{\ell!}{\prod_{i=0}^{\mu} \delta_i(\tilde{\t})!} \rhoprobt(\tilde{\t})^{-1}\right)
\end{align*}
into the intervals
\begin{align*}
\Iset_\t := \left[ \sum_{\substack{\tilde{\t} \in \Tsettord \\ \tilde{\t} < \t}} \tfrac{\ell!}{\prod_{i=0}^{\mu} \delta_i(\tilde{\t})!} \rhoprobt(\tilde{\t})^{-1}, \, \sum_{\substack{\tilde{\t} \in \Tsettord \\ \tilde{\t} \leq \t}} \tfrac{\ell!}{\prod_{i=0}^{\mu} \delta_i(\tilde{\t})!} \rhoprobt(\tilde{\t})^{-1} \right).
\end{align*}
for all $\t \in \Tsettord$.
Then, we draw a random rational number $x$ from $\Iset$.
Since the intervals $\Iset_\t$ form a partition of $\Iset$, there is a unique $\t \in \Tsettord$ with $x \in \Iset_\t$.
As all the interval borders are rational numbers whose denominators divide $\iota$, it follows from the way of choosing $x$, that the probability that $x \in \Iset_\t$ is exactly the ratio of \rev{the lengths of the intervals $\Iset_\t$ and $\Iset$}---hence, $\t$ is drawn from the distribution $\tilde{p}_{\t'}$.

The remaining question is how to determine \rev{which vector $\t$ is such that $x\in I_\t$}.
Algorithm~\ref{alg:draw_efficiently_from_tprime_distribution} computes $\t$ efficiently using a technique similar to enumerative coding \cite{cover1973enumerative}.
The idea is that we iteratively compute for which prefix of $\t$ of length $i$, the real number $x$ is contained in the interval
\begin{align*}
&\Iset_\t^{(i)} = \left[ I_\t^{(i,\mathsf{l})}, I_\t^{(i,\mathsf{r})} \right)
:= \left[ \sum_{\substack{\tilde{\t} \in \Tsettord \\ \tilde{\t} <_i \t}}  \tfrac{\ell!}{\prod_{i=0}^{\mu} \delta_i(\tilde{\t})!} \rhoprobt(\tilde{\t})^{-1}, \, \sum_{\substack{\tilde{\t} \in \Tsettord \\ \tilde{\t} \leq_i \t}}  \tfrac{\ell!}{\prod_{i=0}^{\mu} \delta_i(\tilde{\t})!}\rhoprobt(\tilde{\t})^{-1} \right),
\end{align*}
Note that
\begin{align*}
I_\t^{(i,\mathsf{l})} \leq I_\t^{(j,\mathsf{l})} < I_\t^{(j,\mathsf{r})} \leq I_\t^{(i,\mathsf{r})}
\end{align*}
for all $1\leq i \leq j \leq \ell$ and
\begin{align*}
\Iset_\t = \Iset_\t^{(\ell)}.
\end{align*}
Note that if $x \in \Iset_{[t_1,\dots,t_{i-1}]}^{(i-1)}$, then there is exactly one $t_i$ such that $x \in \Iset_{[t_1,\dots,t_{i}]}^{(i)}$, and we can compute it as
\begin{align}
t_i = \max\left\{ t'' \, : \, \Iset_{[t_1,\dots,t_{i-1},t'']}^{(i,\mathsf{l})} \leq x \right\} \label{eq:drawing_efficiently_from_tprime_distribution_iterative_ti_computation_1}
\end{align}
Furthermore, we have
\begin{align}
\Iset_{[t_1,\dots,t_{i-1},t'']}^{(i,\mathsf{l})} &= \overbrace{\sum_{\substack{\tilde{\t} \in \Tsettord \\ \tilde{\t} <_{i-1} [t_1,\dots,t_{i-1}]}} \tfrac{\ell!}{\prod_{i=0}^{\mu} \delta_i(\tilde{\t})!} \rhoprobt(\tilde{\t})^{-1}}^{= \, \Iset_{[t_1,\dots,t_{i-1}]}^{(i-1,\mathsf{l})}}
+ \sum_{t'=0}^{t_i-1} \sum_{\substack{\tilde{\t} \in \Tset \\ \tilde{\t} =_{i} [t_1,\dots,t_{i-1},t']}} \tfrac{\ell!}{\prod_{i=0}^{\mu} \delta_i(\tilde{\t})!} \rhoprobt(\tilde{\t})^{-1}. \label{eq:drawing_efficiently_from_tprime_distribution_iterative_ti_computation_2}
\end{align}
Equations \eqref{eq:drawing_efficiently_from_tprime_distribution_iterative_ti_computation_1} and \eqref{eq:drawing_efficiently_from_tprime_distribution_iterative_ti_computation_2} combined prove that Lines~\ref{line:draw_efficiently_from_tprime_distribution_draw_uniform_x_from_interval}--\ref{line:draw_efficiently_from_tprime_distribution_sorted_vector_t_obtained} indeed compute the ``index'' $\t$ for which $x \in \Iset_\t$.
This concludes the correctness proof.

The complexity follows since we need to compute $\sum_{\tilde{\t} \in \Tsettord, \tilde{\t} =_{\ell'} \t} \tfrac{\ell!}{\prod_{i=0}^{\mu} \delta_i(\tilde{\t})!} \rhoprobt(\tilde{\t})^{-1}$ for at most $\ell \mu \leq n$ different vectors $\t$, and the cost to compute each of these sums from the precomputed table $\{\Msum(t',\ell',\nmmin',s')\}_{t'\leq t, \ell' \leq \ell}^{\nmmin' \leq \nmmin, s' \leq s}$ as derived in Lemma~\ref{lem:sum_with_prefix}.
The cost of drawing $x$ corresponds to drawing uniformly at random a non-negative integer smaller than $\iota \sum_{\tilde{\t} \in \Tset} \rhoprobt(\tilde{\t})^{-1}$, i.e., of bit size $\in O(n\mu \zeta^2 \log_2(q))$.
This cost, as well as the cost of drawing a random permutation of $\t$, is negligible.
\end{proof}

\subsection{Proof of the Main Statement}\label{ssec:proof_main_statement}

The following proof summarizes the statements shown in this section, which all together imply the main statement, \cref{thm:generic_decoder_simple}.

\begin{proof}[Proof of \cref{thm:generic_decoder_simple}]
First note that for $\eta \leq m$, the algorithm sets $\ambdim = \eta$ and draws a random row support. For $\eta > m$, we do the same for the column case.
Correctness follows since if a suitable $\e$ exists, there is a non-zero probability that a (row or column) super-support of $\e$ is drawn, and erasure decoding has a unique result for a super-support of weight $s <d$ (cf.~\cref{thm:column_erasure_decoding} and \cref{thm:row_erasure_decoding}).

The expected complexity is given by the product of the cost of one iteration $\Wperiteration$ (erasure decoding plus random support drawing) and the expected number of iterations.
The first value can be lower-bounded by $1$ and upper-bounded by $O\!\left(n^3 m^3 \log_2(q)\right)$ due to \cref{thm:column_erasure_decoding}, \cref{thm:row_erasure_decoding}, and \cref{prop:drawing_efficiently_from_tprime_distribution}. The bounds on the expected number of iterations directly follow from \cref{prop:bounds_success_probability} and \cref{prop:simplified_bound_lemma}.

The claim that the bounds \eqref{eq:WF_good_lower_bound} and \eqref{eq:WF_good_upper_bound} can be computed efficiently follows directly from \cref{prop:sum_compute}.
\end{proof}

\section{Comparison to Other Generic Decoders}\label{sec:comparison}

We compare the new generic decoder to other (naive) generic decoding strategies, as well as to existing generic decoders in the extreme cases $\ell=1$ (rank metric) and $\ell=n$ (Hamming metric).

\subsection{Comparison to Extreme Cases: Hamming and Rank Metric}\label{ssec:comparison_Hamming_rank}

In the Hamming-metric case ($\ell=n$), the set $\Tsett$ consists of all permutations of the vector $[1,\dots,1,0,\dots,0]$, where the number of ones equals $t$.
In particular, we have $|\Tsett| = \binom{n}{t}$.
For $\t \in \Tsett$ and $t \leq s \leq n-k$, the function $\scomp(\t,s)$ (\cref{alg:scomp}) returns a random vector $\s \in \{0,1\}^{n}$ with exactly $s$ ones and whose support contains the support of $\t$. In particular, $\rhoprobt(\t) = 1$ for all $\t \in \Tsett$.
Hence, \cref{alg:generic_decoder} uniformly at random selects a subset of $s$ positions in a vector of length $n$, and succeeds if and only if these $s$ positions contain the error positions of an error corresponding an error $\e'$ with $\r -\e' \in \Code$, where $\r$ is the received word.
Although the bounds in \cref{thm:generic_decoder_simple} are---as expected---quite bad for this case (we get $0 \leq W_{\mathrm{new}} \leq \Wperiteration\binom{n}{t}$), this algorithm equals exactly Prange's information-set decoder \cite{prange1962use}, which has expected runtime
\begin{align*}
W_\mathrm{Prange} = \Wperiteration \frac{\binom{n}{t}}{\binom{s}{t}}.
\end{align*}
where $\Wperiteration$ denotes the (polynomial-time) cost of one iteration.

In the rank-metric case ($\ell=1$), the set $\Tsett$ contains only one element: $[t] \in \ZZ^{1}$.
\cref{alg:generic_decoder} thus chooses uniformly at random a row or column space of dimension $s$, and row- or column-erasure decodes in the rank metric.
This method is exactly the rank-syndrome decoder by Gaborit, Ruatta, and Schrek \cite{gaborit2016rsd}.
The complexity bound \eqref{eq:WF_simplified_upper_bound} in \cref{thm:generic_decoder_simple} simplifies to
\begin{align*}
W_\mathrm{GRS} = \Wperiteration q^{t(\min\{n,m\}-s)},
\end{align*}
where $t \leq s \leq \min\!\left\{n-k, \lfloor\tfrac{m}{n}(n-k)\rfloor\right\}$ and $\Wperiteration$ denotes the (polynomial-time) cost of one iteration.
This coincides exactly with Gaborit, Ruatta, and Schrek's complexity bound.

For arbitrary $\ell$ and $t \leq s \leq \min\!\left\{n-k, \lfloor\tfrac{m}{\npr}(n-k)\rfloor\right\}$, the simple upper complexity bound \eqref{eq:WF_simplified_upper_bound} in \cref{thm:generic_decoder_simple} is
\begin{align*}
\WnewUBsimple &= \Wperiteration \tbinom{\ell+t-1}{\ell-1}\gamma_q^\ell q^{t(\ambdim-\frac{s}{\ell})} \\
&\leq \Wperiteration \tbinom{\ell+t-1}{\ell-1}\gamma_q^\ell q^{t \tfrac{\min\{n,\ell m\}-s}{\ell}}.
\end{align*}
For constant $\ell$, the factor $\tbinom{\ell+t-1}{\ell-1}\gamma_q^\ell$ is polynomial in the code length, and can be neglected compared to the exponential term.
Hence, the exponent of the sum-rank-metric generic decoder is roughly a factor $\ell$ smaller than in the rank-metric case ($\ell=1$).
\rev{Note that the bound $\WnewUBsimple$ appears to be a loose approximation of the actual work factor for large $\ell$ (cf. Figure~\ref{fig:example_m=20_n=60}, \ref{fig:example_m=60_n=60}, and \ref{fig:example_m=25_n=60}). Therefore, we refrain from a discussion of $\WnewUBsimple$ for $\ell \in \Omega(n) $ as this does not necessarily give a good intuition about the work factor}.

Overall, the new generic decoding algorithm smoothly interpolates two generic decoding principles known for the extreme cases: Prange's information-set decoder \cite{prange1962use} for the Hamming metric and Gaborit, Ruatta, and Schrek's decoder \cite{gaborit2016rsd} for the rank metric.
The bounds on the work factor in \cref{thm:generic_decoder_simple} are, in a rough sense, good for $\ell$ not too large.
For constant $\ell$, the logarithm of the work factor of our generic $\ell$-sum-rank decoder is roughly a factor $\ell$ smaller than Gaborit, Ruatta, and Schrek's rank-metric decoder.

\subsection{Comparison to Naive Generic Sum-Rank Decoders}\label{ssec:compareNaive}

We compare the new generic decoder to other possible generic decoding strategies.
One naive strategy for generic decoding is given by brute-forcing the codewords, which has a complexity $\Wcode = q^{mk}m^2kn$, where $m^2kn$ is the cost of encoding. Another naive approach is brute-forcing the errors with complexity $\Werror$ as in \eqref{eq:Werror}  (see~\cref{cor:naive_generic_decoder_Werror} in \cref{sec:enumerating_sr_vectors}).
For the extreme cases $\ell=1$ and $\ell=n$, we compare the bounds on the work factor of the new decoder with the Gaborit--Ruatta--Schrek decoder ($\Wnaive$) and Prange's information-set decoder ($W_\mathrm{Prange}$), respectively, cf.~Section~\ref{ssec:comparison_Hamming_rank}.

In Figures~\ref{fig:example_m=20_n=60},~\ref{fig:example_m=60_n=60}, and~\ref{fig:example_m=25_n=60},
we compare the expected complexities of these generic decoding algorithms with the algorithm that we propose. We plot all bounds on the work factor of the new generic decoder that we present in the main statement, \cref{thm:generic_decoder_simple}, as well as the work factor of the ``optimal choice'' for $\ps$ as derived in \cref{app:optimal_support_finding_via_linear_programming}.
We show the $\log_2$ of the number of operations required. %

For all considered parameters, we observe that the difference of the derived upper and lower bound $\WnewLB$ and $\WnewUB$ is small, which indicates that the bounds must be tight on the true work factor. Furthermore, for small values of $\ell$, the simplified upper bound $\WnewUBsimple$ is very close to $\WnewUB$ and becomes loose only for large values of $\ell$. We note that the optimal solution derived in \cref{app:optimal_support_finding_via_linear_programming} is almost exactly on the accurate upper bound $\WnewUB$, for all cases in which we can compute it. Further, the work factor of Prange's algorithm (case $\ell = n$) and the generic rank-metric decoder (case $\ell=1$) are close to the upper bound $\WnewUB$.

\begin{figure}[ht!]
\begin{center}
\includegraphics{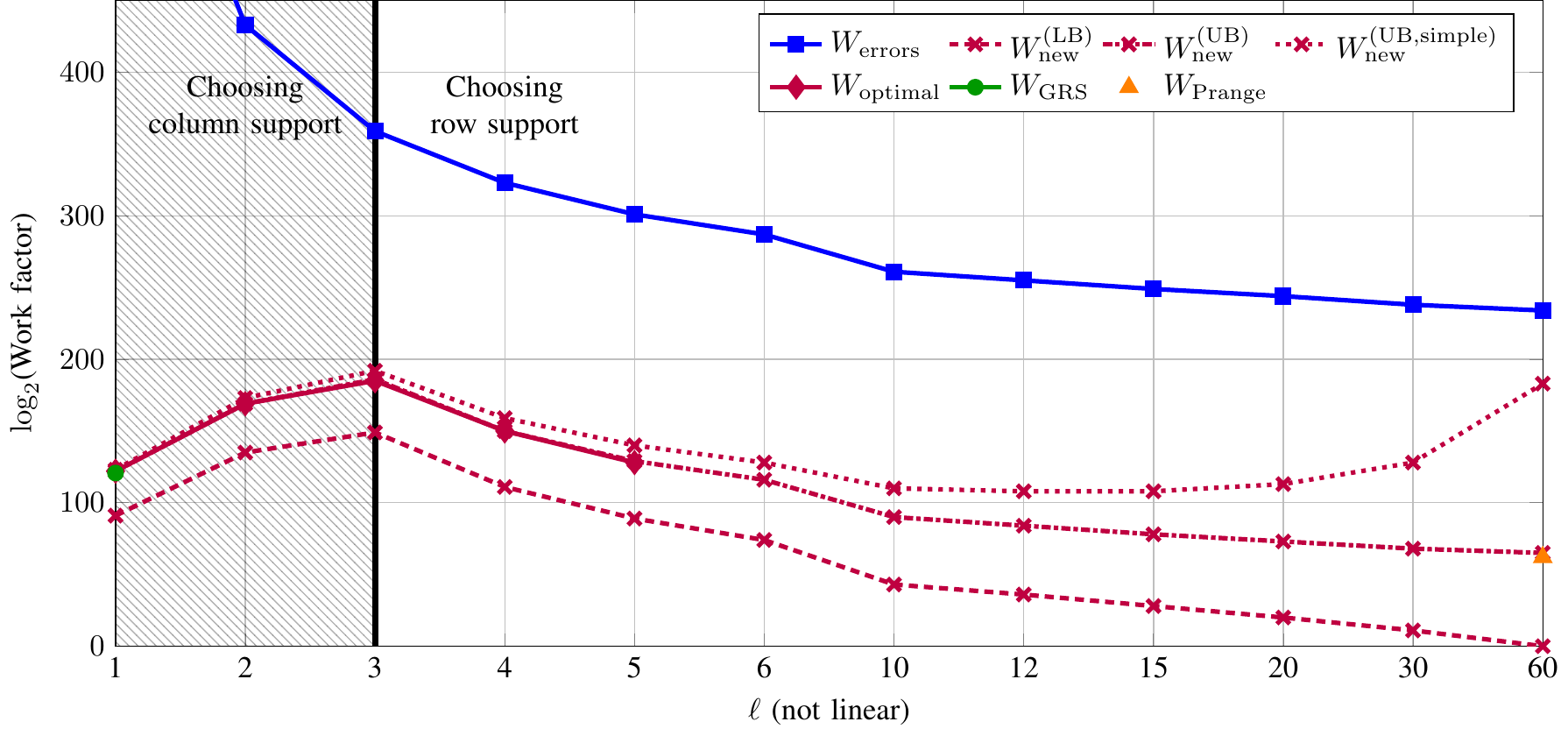}
\end{center}
\caption{Comparison of different generic decoding strategies for $q=2$, $m=20$, $n=60$, $k=30$, $t=9$, $s=10$. The work factor $\Wcode$ is $2^{620}$ for all values of $\ell$ and $\Werror$ is equal to $2^{661}$ for $\ell = 1$.
}
    \label{fig:example_m=20_n=60}
\end{figure}

\begin{figure}[ht!]
\begin{center}
\includegraphics{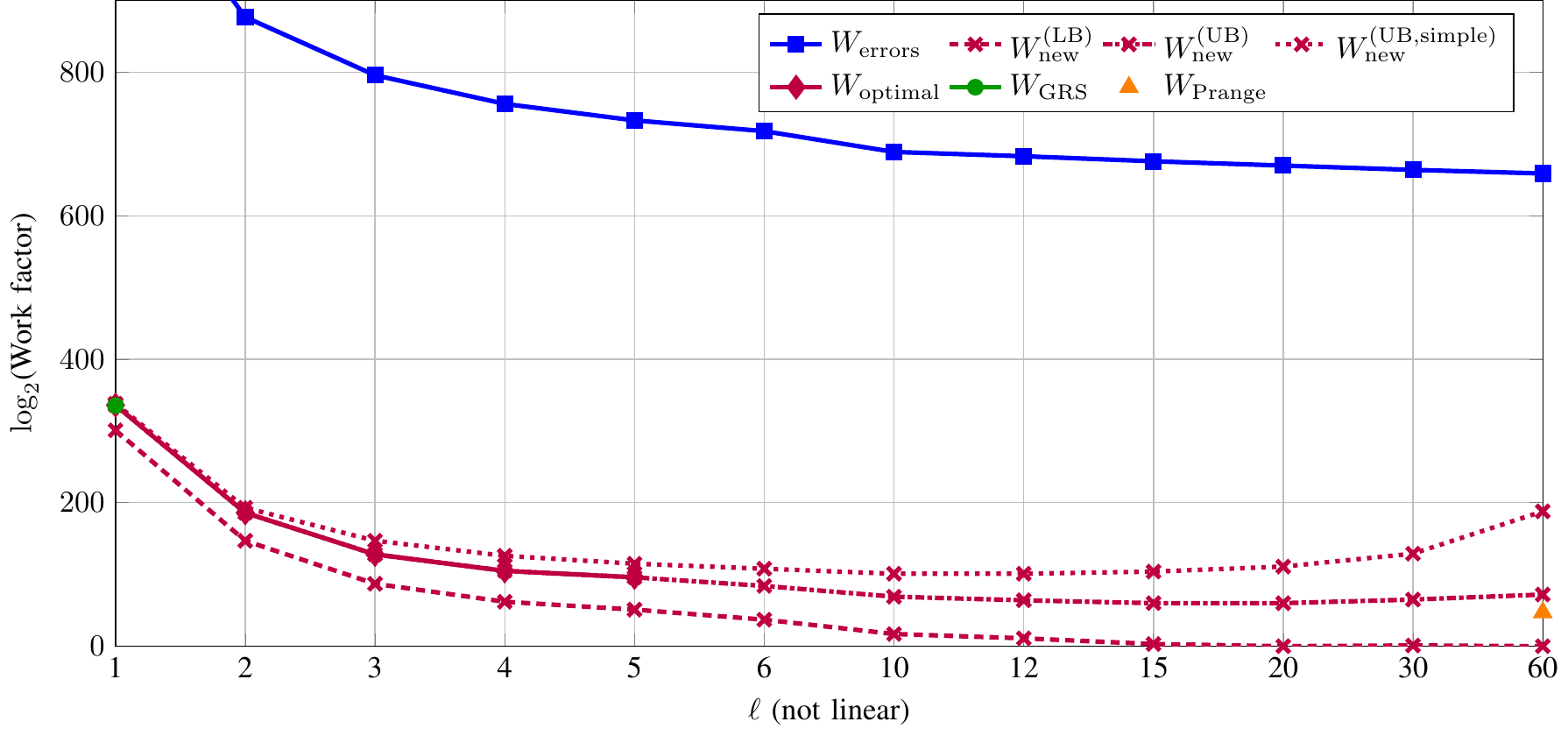}
\end{center}
\caption{Comparison of different generic decoding strategies for $q=2$, $m=60$, $n=60$, $k=30$, $t=10$, $s=30$, where we chose the row support for all values of $\ell$ in the proposed algorithm. The work factor $\Wcode$ is $2^{1823}$ for all values of $\ell$ and $\Werror$ is equal to $2^{1125}$ for $\ell = 1$.
}
\label{fig:example_m=60_n=60}
\end{figure}

\begin{figure}[ht!]
\begin{center}
\includegraphics{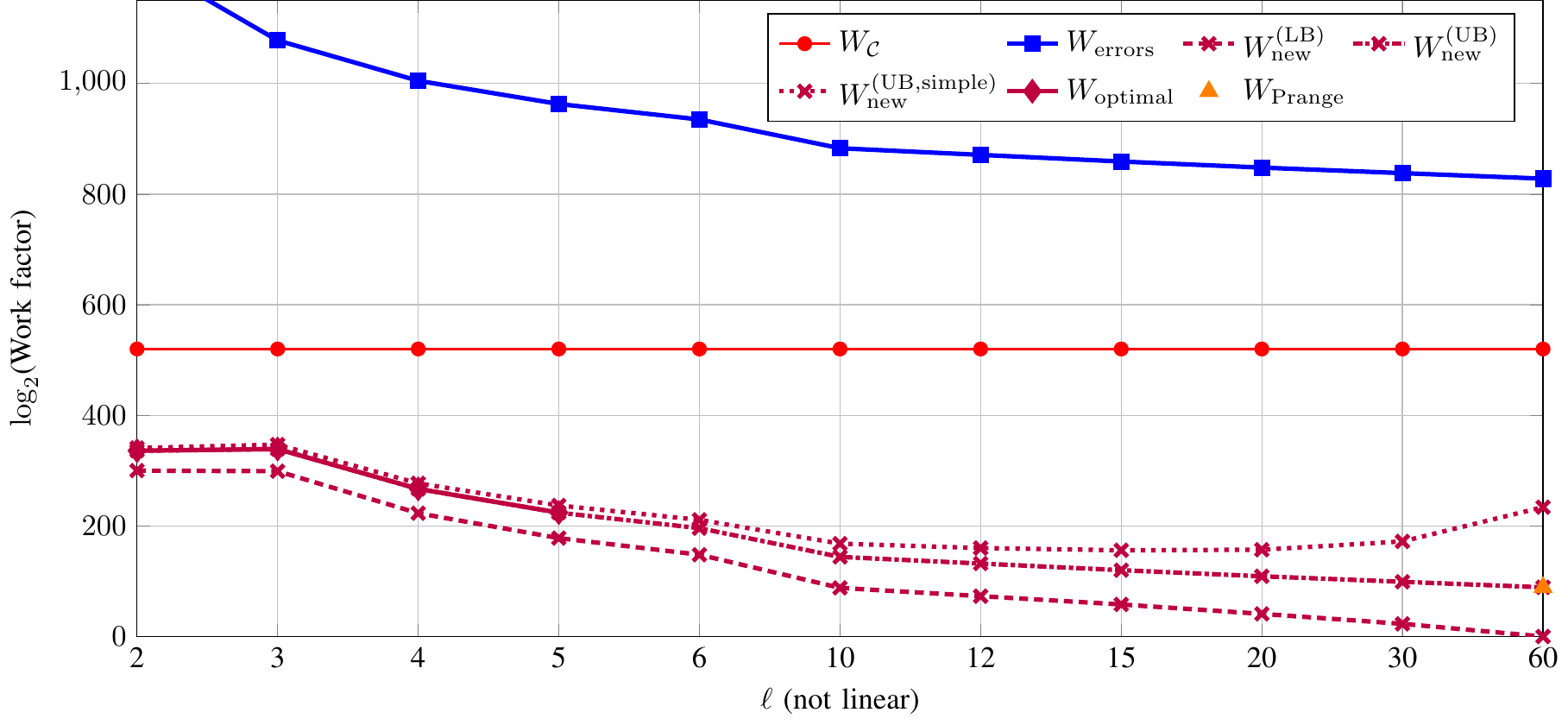}
\end{center}
\vspace{-0.5cm}
\caption{Comparison of different generic decoding strategies for $q=2$, $m=25$, $n=60$, $k=20$, $t=30$, $s=30$. The work factor $\Werror$ is equal to $2^{1225}$ for $\ell = 2$. The case $\ell=1$ is not feasible since the condition $t\leq s\leq\min\!\left\{n-k, \lfloor\tfrac{m}{n}(n-k)\rfloor\right\}$ is not fulfilled.
}
\label{fig:example_m=25_n=60}
\end{figure}

\section{Formal Hardness Proof}\label{sec:hardness}

In this section, we formally prove the hardness of the decisional version of the generic decoding problem in the sum-rank metric.
We adapt the approach of Gaborit and Z\'emor \cite{gaborit2016hard}, who probabilistically reduced the decisional Hamming syndrome decoding problem to the decisional rank syndrome decoding problem over a sufficiently large field extension.
We generalize the method from $\ell=1$ to arbitrary $\ell$, where the size of the extension field can be chosen smaller than in \cite{gaborit2016hard} for $\ell>1$.

\subsection{Complexity Classes}
Let $\mathcal{A}$ be an algorithm that gets as input a sequence of random bits $r$ and the input $x$ of a problem. Then $\mathcal{A}$ is called probabilistic polynomial time ($\PPT$) algorithm if the size of the random sequence $|r|$ (number of bits) is polynomial in the input $|x|$ and $\mathcal{A}$ runs in time polynomial in $|x|$.

{
We make use of the following complexity classes (see, e.g., \cite{trevisan2004complexity}).
Here, $L$ is a decision problem, $0 \leq \const < 1$ is any constant:
\begin{itemize}
\item $L \in \Pclass$ (polynomial time):
  there is a $\PPT$ algorithm $\Ap$ with output $\True$, $\False$ such that $\forall x \in L$ we have $\forall r \,\, \Ap(x,r)=\True$; and $\forall x \not\in L$ we have $\forall r \, \, \Ap(x,r)=\False$.
\item $L \in \RP$ (randomized polynomial-time):
  there is a $\PPT$ algorithm $\Arp$ with output $\True$, $\False$ such that $\forall x \in L$ then $ \Pr[\Arp(x,r)= \True] \geq \const$; and $\forall x \not\in L$ we have $\forall r \, \, \Arp(x,r)=\False$.
  Note that the probability is over the randomness of the bits $r$, and the input $x$ is considered fixed.

  \item $L \in \coRP$ (co-randomized polynomial-time):
  there is a $\PPT$ algorithm $\Acorp$ with output $\True$, $\False$ such that $\forall x \in L$ then $\forall r \,\, \Acorp(x,r)=\True$; and $\forall x \not\in L$ then $\Pr[\Acorp(x,r)=\False] \geq \const$.

  \item $L \in \ZPP$ (zero-error probabilistic polynomial time):
  there is a $\PPT$ algorithm $\Azpp$ with output $\True$, $\False$ or $\mathsf{fail}$ such that the following two are satisfied:
  1) For all $x$ then $\Pr[\Azpp(x,r)=\mathsf{fail}] \leq \const$; and 2) for all $x$ and $r$ then $\Azpp(x,r) = \True \implies x \in L$ and $\Azpp(x, r) = \False \implies x \notin L$.
  Note that $\ZPP = \RP \cap \coRP$.

  \item $L \in \NP$ (non-deterministic polynomial time):
  there is a $\PPT$ algorithm $\Anp$ such that $x \in L$ exactly when there exists an $r$ such that $\Anp(x,r) = \True$.
\end{itemize}

}

We have that $\mathsf{P} \subseteq \ZPP \subseteq \RP \subseteq \NP$.
Assuming that the widely believed conjecture $\ZPP\neq\NP$ was true, then our hardness reduction below would imply that the decisional generic decoding problem in the sum-rank metric was in $\NP \setminus \mathsf{P}$, Hence, it appears likely that the problem is hard to solve.

\subsection{Decoding Problems}

We relate the complexity classes of the following decision problems to each other.

\begin{problem}[Decisional Hamming Syndrome Decoding (\DHSD) Problem]\hfill
\begin{description}	
	\item \textbf{Given:}
\begin{itemize}
\item Parity-check matrix $\H \in \Fq^{(n-k) \times n}$ of a code $\Code$
\item Syndrome $\s \in \Fq^{n-k}$
\item Integer $0 \leq t \leq n$
\end{itemize}
\item \textbf{Question:} Is there an $\e \in \Fq^n$ with $\wtH(\e)\leq t$ such that $\s = \e\H^\top$?
\end{description}
\end{problem}

The \DHSD problem was proven to be $\NP$-complete in \cite{berlekamp1978inherent}.

\begin{problem}[Decisional $\ell$-Sum-Rank Syndrome Decoding (\DSRSD) Problem]\hfill
	\begin{description}
	\item 	 \textbf{Given:}
\begin{itemize}
\item Parameter $\ell \mid n$
\item Parity-check matrix $\H \in \Fqm^{(n-k) \times n}$ of a code $\Code$
\item Syndrome $\s \in \Fqm^{n-k}$
\item Integer $0 \leq t \leq n$
\end{itemize}
\item \textbf{Question:} Is there an $\e \in \Fq^n$ with $\wtSR(\e)\leq t$ such that $\s = \e\H^\top$?
\end{description}

\end{problem}

Note that the \DSRSDnoell{n} problem and the \DHSD problem are the same. The \DSRSDnoell{1} problem is the decisional rank-syndrome decoding problem, which was shown to be hard in the following way \cite{gaborit2016hard}: If the \DSRSDnoell{1} problem is in $\ZPP = \RP \cap \coRP$, then $\NP=\ZPP$. The next subsection generalizes this statement to arbitrary $\ell$.

\subsection{Hardness Reduction}

The following statements constitute the formal hardness proof, which is summarized in \cref{thm:hardness_main_statement}.
The proof strategy is similar to the proof of the probabilistic reduction of the ``decisional minimum rank weight problem'' in \cite{gaborit2016hard}
(note that Gaborit and Z\'emor prove the reduction for the \DSRSDnoell{1} by referring to the analogy to the latter problem).
Compared to the original statement in the case $\ell=1$, we can improve the tightness of the reduction (for $\ell>1$) using the bound on the sum-rank-metric sphere size derived in \cref{thm:bound_number_fixed-weight-words} in \cref{sec:enumerating_sr_vectors}.
We start with a technical lemma, which we will use to bound some probabilities in our probabilistic reductions.

\begin{lemma}\label{lem:hardness_probability}
Let $\varepsilon>0$ be fixed and $m,n,\ell$ be positive integers with $m \geq \tfrac{n^2}{\ell}+n\log_q(8n)+ \log_q(2\varepsilon^{-1})$.
Let $\H \in \Fq^{(n-k)\times n}$, $\s \in \Fq^{n-k}$ and $\x \in \Fq^n$, where $\x$ is a vector of minimum Hamming weight $\tH$ such that $\x \H^{\top} = \s$. Further let $\ve{\beta}$ be chosen uniformly at random from $(\Fqm^\ast)^n$ %
and let then $\x' \in \Fqm^n$ be a vector of minimum sum-rank weight such that $\x' \big(\H \diag(\ve{\beta})\big)^{\top} = \s$.
Then, the probability that $\wtSR(\x') < \tH$ is at most $\varepsilon$.
\end{lemma}

\begin{IEEEproof}
Let $\H$, $\s$ and $\tH$ be fixed.
We define $P$ as the probability (randomness in $\ve{\beta}$)
\begin{align*}
 P := \Pr\left\{ \exists \x' \in \Fqm^n \, : \, \x' \big(\H \diag(\ve{\beta})\big)^{\top} = \s \, \land\, \wtSR(\x') < \tH \right\}.
\end{align*}

For randomly chosen $\ve{\beta}\sample (\Fqm^\ast)^n$, let $\mathcal{E}_{\a}$ be the event that for a fixed vector $\a \in \Fqm^n$, the equality $\a  \big(\H \diag(\ve{\beta})\big)^{\top} = \s$ holds.
Further define the set $\mathcal{X}(\tH-1) := \{\a \in \Fqm^n: \wtSR(\a)< \tH\}$. Then,
\begin{align*}
\Pr \left\{ \exists \x' \in \Fqm^n: \x'  \big(\H \diag(\ve{\beta})\big)^{\top} = \s \, \land \, \wtSR(\x') < \tH \right\} = \Pr \Bigg[ \bigcup_{\x' \in \mathcal{X}(\tH)}  \mathcal{E}_{\x'} \Bigg] \leq \sum_{\x' \in \mathcal{X}(\tH)} \Pr [ \mathcal{E}_{\x'} ].
\end{align*}

Next, we bound $\Pr [ \mathcal{E}_{\x'}]$ for a given $\x' \in \Fqm^n$.
If there exists a $\ve{\beta} \in (\Fqm^\ast)^n$ such that $\x'  \big(\H \diag(\ve{\beta})\big)^{\top} = \s$, then, by \cite[Lemma 4]{gaborit2016hard}, there exists a subset $\mathcal{W} \subseteq \{i \, : \, x_i' \neq 0\}$ of cardinality $|\mathcal{W}|=\tH$ such that the columns $\h_i$ of $\H$ indexed by $i \in \mathcal{W}$ are linearly independent.
Hence, if we fix $\beta_i$ for all $i \notin \mathcal{W}$, then the set of vectors $\x'  \big(\H \diag(\ve{\beta})\big)^\top$ obtained by choosing the remaining $\beta_i \in \Fqm^\ast$ for $i\in\mathcal{W}$ has cardinality $(q^m-1)^{\tH}$.
Hence, for $\ve{\beta}\sample (\Fqm^\ast)^n$, we have
\begin{align*}
\Pr [ \mathcal{E}_{\x'} ] \leq \frac{1}{(q^m-1)^{\tH}}.
\end{align*}
Otherwise, if there is no $\ve{\beta} \in (\Fqm^\ast)^n$ such that $\x' \big(\H \diag(\ve{\beta})\big)^{\top} = \s$, then we have $\Pr [ \mathcal{E}_{\x'} ] = 0$, which is obviously $\leq \frac{1}{(q^m-1)^{\tH}}$.

Define $\Gamma(q,m,\tH) := \frac{q^{m\tH}}{(q^m-1)^{\tH}} = \frac{1}{(1-q^{-m})^{\tH}}$. Since $m \geq \tH$ by assumption, we have
\begin{align*}
\Gamma(q,m,\tH) \leq \frac{1}{(1-q^{-m})^{m}} = \frac{1}{\sum_{i=0}^{m} \binom{m}{i}(-q^{-m})^i} \overset{(\ast)}{\leq} \frac{1}{1-m q^{-m}} \leq 2,
\end{align*}
where we use $mq^{-m} \leq \tfrac{1}{2}$, and $(\ast)$ follows from the fact that the terms in the sum in the second line are alternating and their absolute values are strictly monotonically decreasing, i.e.,
\begin{align*}
\binom{m}{i+1}q^{-(i+1) m} = q^{-m}\tfrac{m-i}{i+1} \binom{m}{i}q^{-i m} \leq \tfrac{1}{2}\binom{m}{i}q^{-i m}.
\end{align*}
Combining the arguments above, we get
\begin{align*}
  P &\leq \frac{1}{(q^m-1)^{\tH}} |\mathcal{X}(\tH-1)|\\
    &= \Gamma(q,m,\tH)\frac{1}{q^{m\tH}} \sum_{i=1}^{\tH-1} \mathcal{N}_{q,\npr,m}(i,\ell)\\
    &\leq 2\frac{1}{q^{m\tH}} (\tH-1) \max_{i \in [1,\tH-1]}\mathcal{N}_{q,\npr,m}(i,\ell) \\
   &\leq 2\frac{1}{q^{m\tH}} \frac{1}{q^{m\tH}}(\tH-1) \binom{\ell+\tH-2}{\ell-1} 4^\ell  q^{(\tH-1)(m+\npr-\frac{\tH-1}{\ell})} \\
    &= 2(\tH-1) \binom{\ell+\tH-2}{\ell-1} 4^\ell  q^{-m+(\tH-1){\npr}-\frac{(\tH-1)^2}{\ell}} \\
    &\leq 2(\tH-1) \binom{\ell+\tH-2}{\ell-1} 4^\ell  q^{-m+\frac{n^2}{\ell}-\frac{(\tH-1)^2}{\ell}}. \\
    &\leq 2\underbrace{(\tH-1)}_{\leq \, \ell+\tH-2} (\ell+\tH-2)^{\ell-1} 4^\ell  q^{-m+\frac{n^2}{\ell}-\frac{(\tH-1)^2}{\ell}} \\
	&\leq 2[4(\ell+\tH-2)]^{\ell}  q^{-m+\frac{n^2}{\ell}-\frac{(\tH-1)^2}{\ell}} \\
	&\leq 2q^{-m+\frac{n^2}{\ell}-\frac{(\tH-1)^2}{\ell} + \ell \log_{q}[4(\ell+\tH-2)]} \\
	&\leq 2q^{-m+\frac{n^2}{\ell} + \ell \log_{q}[4(\ell+\tH-2)]} \\
	&\leq 2q^{-m+\frac{n^2}{\ell} + n \log_{q}(8n)} \\
	&\leq \varepsilon. \qedhere
\end{align*}
\end{IEEEproof}

We first show, that if there is a $\coRP$-algorithm for $\DSRSD$, then we can make a $\coRP$-algorithm for $\DHSD$ (\cref{alg:ahco} below).
  The idea is simple: The algorithm transforms the input into an instance of the $\DSRSD$ problem via a random linear map, and simply calls the $\coRP$-algorithm for $\DSRSD$.
Using \cref{lem:hardness_probability}, we can show that the solution to this problem will usually project back to a solution to the $\DSRSD$ instance.

\begin{lemma}\label{lem:hardness_coRP}
For any $\ell<n$ and $m > \tfrac{n^2}{\ell}+n\log_q(8n)+ \log_q(2)$, if the \DSRSD problem is in $\coRP$, then the \DHSD problem is in $\coRP$.
\end{lemma}
\begin{IEEEproof}
  Let $\Arco$ be a hypothesised $\coRP$-algorithm for the \DSRSD problem, i.e.~it inputs an instance $(\H'\in\Fqm^{(n-k)\times n}, \s\in\Fqm^{n-k}, t \in \ZZ_{>0})$ and outputs $\True$ whenever $\tSR \leq t$, while it outputs $\False$ with probability at least $1 - \tilde \varepsilon$ if $\tSR > t$, where $\tSR$ is the minimum sum-rank weight of the vectors $\x'\in\Fqm^{n}$ such that $\x' \H'^{\top} = \s$, and $\tilde \varepsilon \geq 0$ is some fixed constant.

  Then \cref{alg:ahco} details an $\coRP$-algorithm $\Ahco$ for the \DHSD problem that inputs an instance $(\H \in \Fq^{(n-k)\times n}, \s\in\Fq^{n-k}, t \in \ZZ_{> 0})$.
  We should show that $\Ahco$ outputs $\True$ whenever $\tH \leq t$, while it outputs $\False$ with at least some non-zero constant probability if $\tH > t$, where $\tH$ denotes the minimum Hamming weight of the vectors $\x \in \Fq^n$ such that $\x \H^{\top} = \s$.

  Observe first that if $\tH \leq t$, it follows that $\tSR \leq t$, so $\Ahco$ outputs $\True$.
  Consider now the case $\tH > t$.
  By the definition of $m$, we may choose a non-negative constant $\varepsilon < 1$ such that $m \geq \tfrac{n^2}{\ell}+n\log_q(8n)+ \log_q(2\varepsilon^{-1})$.
  Hence by \cref{lem:hardness_probability}, with probability $\geq 1-\varepsilon$, we have $\tSR = \tH > t$, and so $\Ahco$ outputs $\False$ with probability at least $(1 - \varepsilon)(1-\tilde \varepsilon)$, which is again a constant.
\end{IEEEproof}

\begin{algorithm}
	\caption{$\Ahco$}\label{alg:ahco}
	\SetKwInOut{Input}{Input}\SetKwInOut{Output}{Output}
	\Input{$\H \in \Fq^{(n-k)\times n}$, $\s \in \Fq^{n-k}$, integer $t$}
	\Output{$\True$ or $\False$} %
	$\ve{\beta} \sample (\Fqm^\ast)^n$ \\
  $\H' \gets \H \diag(\ve{\beta}) \in \Fqm^{(n-k)\times n}$ \\
  \Return $\Arco(\H',\s)$
\end{algorithm}

The $\RP$ reduction is more involved: in order to solve the $\DHSD$ problem, we solve its related search problem by calling a (hypothetical) $\RP$-algorithm for $\DSRSD$ at most $n$ times on certain punctured and randomly transformed parity-check matrices.

\begin{lemma}\label{lem:hardness_RP}
For any $\ell<n$ and $m \geq \tfrac{n^2}{\ell}+n\log_q(8n)+ \log_q(4n)$, if the \DSRSD problem is in $\RP$, then the \DHSD problem is in $\RP$.
\end{lemma}
\begin{IEEEproof}
Let $\Arrp$ be a hypothesised $\RP$-algorithm for $\DSRSD$, i.e.~it inputs an instance $(\H'\in \Fqm^{(n-k)\times n'}, \s \in \Fq^{n-k}, t \in \ZZ_{>0})$, and outputs $\False$ whenever $\tSR > t$ and outputs $\True$ with probability $1 - \tilde\varepsilon$ if
$\tSR \leq t$, where $\tSR$ is the minimum sum-rank weight of the vectors $\x'\in\Fqm^{n'}$ such that $\x' \H'^{\top} = \s$, and $\tilde\varepsilon > 0$ is some constant \rev{smaller} than $1$.
By iterating $\Arrp$ at most $O(\log n)$ times, we may assume $\tilde\varepsilon < \frac 1 {2n}$.

Then \cref{alg:ahrp} details an $\RP$-algorithm $\Ahrp$ for the \DHSD problem that inputs an instance $(\H \in \Fq^{(n-k)\times n}, \s\in\Fq^{n-k}, t \in \ZZ_{> 0})$.
We should show that $\Ahrp$ outputs $\False$ whenever $\tH > t$, while it outputs $\True$ with at least some constant non-zero probability if $\tH \leq t$, where $\tH$ denotes the minimum Hamming weight of the vectors $\x \in \Fq^n$ such that $\x \H^{\top} = \s$.

The idea of the algorithm is to determine a Hamming super-support $\mathcal{S}$ of cardinality at most $t$ of a vector $\x \in \Fq^n$ such that $\x \H^{\top} = \s$.
The function $\col(\H, \mathcal{T})$ returns the sub-matrix of $\H$ consisting of the columns indexed by the index set $\mathcal T$.
Note that each line runs in polynomial time: in particular, \cref{algahrp:solvex} is simply solving a linear system.
From Lines~\ref{algahrp:line14}--\ref{algahrp:line16}, we observe that the algorithm outputs $\True$ whenever a super-support is found, and outputs $\False$ otherwise.
Hence $\Ahrp$ outputs $\False$ whenever $\tH > t$, and we need to show that if $\tH \leq t$ then $\Ahrp$ returns $\True$ with some non-zero constant probability.

So assume $\tH \leq t$. The purpose of Lines~\ref{algahrp:col}--\ref{algahrp:line12} is to answer the following question:
\begin{center}
$\mathrm{(Q)}$ Is $\mathcal{S} \setminus \{i\}$ a super-support of a vector $\x \in \Fq$ with Hamming weight $\wtH(\x) \leq t$ and syndrome $\s = \H \x^\top$? %
\end{center}
Since we start with $\mathcal{S} = \{1,\dots,n\}$, it is clear that if we always get the correct answer to this question, at termination, the set $\mathcal{S}$ will be the support of a vector $\x$ of Hamming weight $\wtH(\x) \leq t$ and syndrome $\s = \H \x^\top$. If we get an incorrect answer in only one of the $\leq n$ loops, then we are not guaranteed that $\mathcal{S}$ has this property, but we can detect this event by Lines~\ref{algahrp:line14}--\ref{algahrp:line16}.

We show that the probability that Lines~\ref{algahrp:col}--\ref{algahrp:line12} answer the question $\mathrm{(Q)}$ correctly in \emph{all} iterations of the loop is at least a constant.
Note that there are two types of randomness in these lines, which both can influence the answer that we get: the choice of $\ve{\beta}$ and the randomness in the algorithm $\Arrp$.
We distinguish two cases and denote for given $\s \in \Fq^{n-k}$, $\mathcal{S}$, $i$, $\ve{\beta} \in \Fqm^{|\mathcal{S}|-1}$, the smallest Hamming weight of a vector $\tilde{\x} \in \Fq^{|\mathcal{S}|-1}$ such that $\tilde{\x} \bar{\H} = \s$ by $\tH$ and the smallest $\ell$-sum-rank weight of a vector $\tilde{\x}' \in \Fqm^{|\mathcal{S}|-1}$ with $\tilde{\x}' \bar{\H}' = \s$ as $\tSRtilde$. Note that the answer to $\mathrm{(Q)}$ is true if and only if $\tHtilde \leq t$.
\begin{itemize}
\item Case 1: The answer to $\mathrm{(Q)}$ is $\True$ (i.e., $\tHtilde\leq t$):
Independent of how $\ve{\beta}$ is chosen, we always have $\tSRtilde\leq \tHtilde \leq t$, so $\Arrp(\bar{\H}',\s,t)$ returns $\True$ (the correct answer) with probability at least $1-\tilde{\varepsilon}> 1-\tfrac{1}{2n}$  (randomness in $\Arrp$) and $\False$ (the incorrect answer) with probability at most $\tilde{\varepsilon}<\tfrac{1}{2n}$.
\item Case 2: The answer to $\mathrm{(Q)}$ is $\False$ (i.e., $\tHtilde> t$):
\begin{itemize}
\item With probability $>1-\tfrac{1}{2n}$ (randomness in the choice of $\ve{\beta}$), the vector $\ve{\beta}$ is chosen such that $\tSRtilde = \tHtilde$ due to Lemma~\ref{lem:hardness_probability} where we set $\varepsilon=\tfrac{1}{2n}$, which is permissible with our restriction on $m$.
In this case, we thus have $\tSRtilde>t$, and $\Arrp(\bar{\H}',\s,t)$ outputs always $\False$ (the true answer).
\item The counter-event of the above occurs with probability $<\tfrac{1}{2n}$: the vector $\ve{\beta}$ is chosen such that $\tSRtilde \leq t < \tHtilde$. In this case, $\Arrp(\bar{\H}',\s,t)$ may return $\True$ (the wrong answer) or $\False$ (the correct answer).
\end{itemize}
\end{itemize}
Hence, in both cases, Lines~\ref{algahrp:col}--\ref{algahrp:line12} answer the question $\mathrm{(Q)}$ correctly with probability greater than $1-\tfrac{1}{2n}$. Since the question is asked at most $n$ times, we get the correct answer to $\mathrm{(Q)}$ in \emph{all} iterations with probability at least $1-\tfrac{n}{2n} = \tfrac{1}{2}$ by the union bound.
\end{IEEEproof}

\begin{algorithm}
	\caption{$\Ahrp$}\label{alg:ahrp}
	\SetKwInOut{Input}{Input}\SetKwInOut{Output}{Output}
	\Input{$\H \in \Fq^{(n-k)\times n}$, $\s \in \Fq^{n-k}$, integer $t$}
	\Output{$\True$ or $\False$}
  $\mathcal{S} = \{1,\hdots,n\} $ \label{algahrp:line6}\\
  \For{$i =1,\hdots,n$ }{
    $\bar{\H} \gets \col(\H, \mathcal{S} \setminus \{i\})\in \Fq^{(n-k)\times (|\mathcal{S}|-1)}$ \label{algahrp:col}\\
    $\ve{\beta} \sample (\Fqm^\ast)^{|\mathcal{S}|-1}$ \\
    $\bar{\H}' \gets \bar{\H} \diag(\ve{\beta}) \in \Fqm^{(n-k)\times (|\mathcal{S}|-1)}$ \\
    \If{$\Arrp(\bar{\H}',\s, t) = \True$}{
      $\mathcal{S} \gets \mathcal{S} \setminus \{i\}$ \\
      \label{algahrp:line12}
    }
  }
  $\bar{\H} \gets \col(\H, \mathcal{S})\in \Fq^{(n-k)\times |\mathcal{S}|}$\\
  \If{$1\leq |\mathcal{S}|\leq t$ \normalfont{\textbf{and}} $\exists \x \in \Fq^{|\mathcal{S}|}$ s.t. $\x\bar{\H}^{\top} = \s$ \label{algahrp:line14}}{ \label{algahrp:solvex}
    \Return{$\True$} \label{algahrp:line15}
  }
  \Else{
    \Return{$\False$}\label{algahrp:line16}
  }
\end{algorithm}

The lemmas above imply the main statement of this section.

\begin{theorem}\label{thm:hardness_main_statement}
For $\ell<n$ and $m \geq \tfrac{n^2}{\ell}+n\log_q(8n)+ \log_q(4n)$, if the \DSRSD problem is in $\ZPP=\RP \cap \coRP$, then we have $\NP= \ZPP$.
\end{theorem}

\begin{proof}
It is well-known that $\ZPP \subseteq \NP$.
The other inclusion, $\ZPP \supseteq \NP$, follows from the NP-hardness of \DHSD, \cref{lem:hardness_coRP}, and \cref{lem:hardness_RP}.
\end{proof}

\begin{remark}
In the special case of \cref{thm:hardness_main_statement} for the rank metric ($\ell=1$), which was shown in \cite{gaborit2016hard}, the restriction on the extension degree is $m > n^2$.
It can be seen that our assumption, $m \geq \tfrac{n^2}{\ell}+n\log_q(8n)+ \log_q(4n)$, is less restrictive for $\ell>1$.
An interesting special case is $\ell \in \Omega(n)$, i.e.~a sum-rank metric close to the Hamming metric, for which we can choose $m \in O(n \log(n))$.
\end{remark}

\section{Conclusion}

We have proposed the first generic decoder in the ($\ell$-)sum-rank metric, which combines
known generic decoding algorithms in the Hamming metric ($\ell=n$) and rank metric ($\ell=1$).
For $\ell=n$, the algorithm resembles the information-set decoder by Prange \cite{prange1962use} and for $\ell=1$, it coincides with the generic decoder for the rank metric by Gaborit, Ruatta, and Schrek \cite{gaborit2016rsd}.

We have derived lower and upper bounds on the runtime of our generic decoding algorithm, which can be computed in small-degree polynomial time in the code parameters.
Furthermore, we derived a simple upper bound on the complexity of the new generic decoding algorithm.
For a constant number of blocks $\ell$, the bound shows that the exponent of our algorithm's work factor is roughly a factor $\ell$ smaller than for the generic rank-metric decoder by Gaborit, Ruatta, and Schrek \cite{gaborit2016rsd}.
Our formal hardness proof in \cref{sec:hardness} extends a result by Gaborit and Z\'emor \cite{gaborit2016hard} from the rank metric, and provides evidence that generic decoding in the sum-rank metric is a hard problem.

Besides being of theoretical interest, the results open up the possibility to study sum-rank-metric codes in code-based cryptosystems.
We have also derived results on the cardinality of sum-rank-metric spheres, which can, among others, be used to efficiently compute bounds on code parameters (cf.~\cref{rem:further_applications_sphere_size}).
Furthermore, the new notion of column support and the erasure decoding algorithms can be of more general interest.

The article can be seen as a proof-of-concept that ideas for generic decoding in the extreme cases, Hamming and rank metric, can be adapted to the family of sum-rank metrics.
An obvious open problem is the study of the many improvements of \cite{prange1962use} in the Hamming and \cite{gaborit2016rsd} in the rank metric.
In particular, it would be interesting to adapt the very recent significant improvement of generic decoding in the rank metric based on algebraic methods \cite{bardet2020algebraic} to the sum-rank metric.
As for the rank metric, it is an open problem whether there is a deterministic reduction from an NP-hard problem to the decisional sum-rank syndrome decoding problem.

\bibliographystyle{IEEEtran}
\bibliography{main}

\appendix
\rev{
\subsection{Generating Uniformly at Random Errors of a Given Sum-Rank Weight}\label{app:drawerrors}

\begin{algorithm}[ht!]
	\caption{$\textsf{Drawing Uniformly at Random an Error of Given Sum-Rank Weight}$}\label{alg:uniform_drawing_sum-rank_errors}
	\SetKwInOut{Input}{Input}\SetKwInOut{Output}{Output}
	\Input{Parameters $q,m,k,n,\ell,t$
	}
	\Output{Vector $\e \sample \{\e' \in \Fqm^n: \wtSR(\e') = t\}$}
	$D^{(1)} \sample \{1,\hdots,\mathcal{N}_{q,\npr,m}(t,\ell) \} $ \label{line:drawD}\\
	$t^{(1)} \gets t$ \label{line:mappingStart}\\
	\For{$j\in\{1,\hdots,\ell\}$}
	{
		$t_j \gets \max\left\{ t'' \in \left\{0,\hdots,t^{(j)}\right\}: \sum_{t'=t^{(j)}-\nmmin(\ell-j)}^{{t''-1}} \NM{q}{m,\npr,t'} \cdot \mathcal{N}_{q,\npr,m}({t^{(j)}-t'},\ell-j) < D^{(j)}\right\} $\\
		$D^{(j+1)} \gets D^{(j)} - \sum_{t'=t^{(j)}-\nmmin(\ell-j)}^{{t_j-1}} \NM{q}{m,\npr,t'} \cdot \mathcal{N}_{q,\npr,m}({t^{(j)}-t'},\ell-j)$ \\
		$t^{(j+1)} \gets t^{(j)} - t_j$\label{line:mappingStop}}
	\For{$j\in\{1,\hdots,\ell\}$\label{line:drawingStart} }
	{
		$\a_j \sample\{\a \in \Fqm^{t_j} : \rk_{\Fq}\!\left(\a\right) = t_j\}$ \\
		$\B_j \sample\{\B \in \Fq^{t_j \times \npr} : \rk_{\Fq} \!\left(\B\right) = t_j\}$	
	}
	$\e \gets [\a_1\B_1 \mid \a_2\B_2 \mid \dots \mid \a_{\ell}\B_{\ell}] \in \Fqm^n$ \label{line:drawingStop} \\
	\Return{$\e$}
\end{algorithm}

The recursion in Lemma~\ref{lem:number_fixed-weight-words_recursion} can be turned into a variant of enumerative coding \cite{cover1973enumerative} to efficiently draw uniformly at random from the set of sum-rank vectors of weight $t$. Such an algorithm is outlined in \cref{alg:uniform_drawing_sum-rank_errors}, and its correctness is proven in the following proposition:
\begin{proposition}
	Let $q,m,k,n,\ell$, and $t$ be integers such that $\ell \mid n$ and $t\leq \nmmin\ell$. Then, \cref{alg:uniform_drawing_sum-rank_errors} outputs a vector $\e\in\Fqm^n$ drawn uniformly at random from $\{\e' \in \Fqm^n: \wtSR(\e') = t\}$. 
\end{proposition}
\begin{proof}
	The set $\{\e' \in \Fqm^n: \wtSR(\e') = t\}$ %
	has cardinality $\mathcal{N}_{q,\npr,m}(t,\ell)$.
	Let $\varphi: \{1,\hdots,\mathcal{N}_{q,\npr,m}(t,\ell)\} \rightarrow  \{\e' \in \Fqm^n: \wtSR(\e') = t\}$ be a bijective mapping. If we know an efficient algorithm to realize the mapping $\varphi$, then the drawing could be simply realized by choosing uniformly at random $D^{(1)}$ from $\{1,\hdots,\mathcal{N}_{q,\npr,m}(t,\ell) \}$ and outputting $\varphi(D^{(1)})$. However, the drawing algorithm can also be realized with a different method. 
		
	Let $\phi: \{\e \in \Fqm^n: \wtSR(\e) = t\} \rightarrow \Tsett$, $\e \mapsto [\rk_{\Fq}(\e_1),\hdots,\rk_{\Fq}(\e_\ell)]$. Then, the drawing can be conducted by computing $\t = (\phi \circ \varphi) (D^{(1)})$ and sampling $\a_j \sample\{\a \in \Fqm^{t_j} : \rk_{\Fq}\!\left(\a\right) = t_j\}$ and
	$\B_j \sample\{\B \in \Fq^{t_j \times \npr} : \rk_{\Fq} \!\left(\B\right) = t_j\}$, for $j\in\{1,\hdots,\ell\}$. Since $\e_j = \a_j \B_j \in \Fqm^{\npr}$ is a vector drawn uniformly at random from $\{\e' \in \Fqm^{\npr}: \rk_{\Fq}(\e') = t_j\}$, it follows that $\e = [\a_1\B_1 \mid \dots \mid \a_{\ell}\B_{\ell}]$ is a vector drawn uniformly at random from $\{\e' \in \Fqm^n: \wtSR(\e') = t\}$.
		
	To derive the mapping $\phi \circ \varphi  : \{1,\hdots,\mathcal{N}_{q,\npr,m}(t,\ell)\} \rightarrow \Tsett$ suppose that $t\leq \nmmin$. Then, the number of vectors that have a weight decomposition $[0,\hdots,0,t]$ is equal to $\NM{q}{m,\npr,t}$, and therefore, we map
	\begin{equation*}
	D\in\{1,\hdots,\NM{q}{m,\npr,t}\} \mapsto [0,\hdots,0,t].
	\end{equation*}
	Furthermore, the number of vectors that have a weight decomposition $[0,\hdots,0,1,t-1]$ is equal to $\NM{q}{m,\npr,1}\NM{q}{m,\npr,t-1}$, which means that we map
	\begin{align*}
	D\in\{\NM{q}{m,\npr,t}+1,\hdots,\NM{q}{m,\npr,t}+\NM{q}{m,\npr,1}\NM{q}{m,\npr,t-1}\}
	\mapsto [0,\hdots,1,t-1].
	\end{align*}	
	It follows by induction that we map
	\begin{equation*}
	D\in\left\{
		\sum_{t'=0}^{\rev{t_j-1}} \NM{q}{m,\npr,t'} \cdot \mathcal{N}_{q,\npr,m}(\rev{t-t'},\ell-j)+1
		,\hdots,
		\sum_{t'=0}^{\rev{t_j}} \NM{q}{m,\npr,t'} \cdot \mathcal{N}_{q,\npr,m}(\rev{t-t'},\ell-j)
		\right\} \mapsto [0,\hdots,0,t_{j},\hdots,t_{\ell}],
		\end{equation*}
	where $\sum_{i=j+1}^{\ell} t_{i} = t-t_j$. 
		
	\cref{alg:uniform_drawing_sum-rank_errors} performs this routine. In \cref{line:drawD}, the integer $D^{(1)}$ is drawn uniformly at random from $\{1,\hdots,\mathcal{N}_{q,\npr,m}(t,\ell) \}$, and in Lines~\ref{line:mappingStart} to \ref{line:mappingStop}, the respective weight distribution vector $(\phi \circ \varphi) (D^{(1)})$ is determined (the cases of $t > \nmmin$ are taken into account by starting to sum from $t^{(j)}-\nmmin(\ell-j)$ instead of $0$). The method to compute $(\phi \circ \varphi) (D^{(1)})$ is illustrated in Figure~\ref{fig:drawingerrors}. In Lines~\ref{line:drawingStart} to \ref{line:drawingStop}, the vectors $\e_j\in\Fqm^{\npr}$ are drawn uniformly at random from the set of vectors of rank weight $t_j$, and the vector $[\e_1 \mid \hdots \mid \e_{\ell}]$ is returned.
\end{proof}

\begin{figure}
	\center
\includegraphics{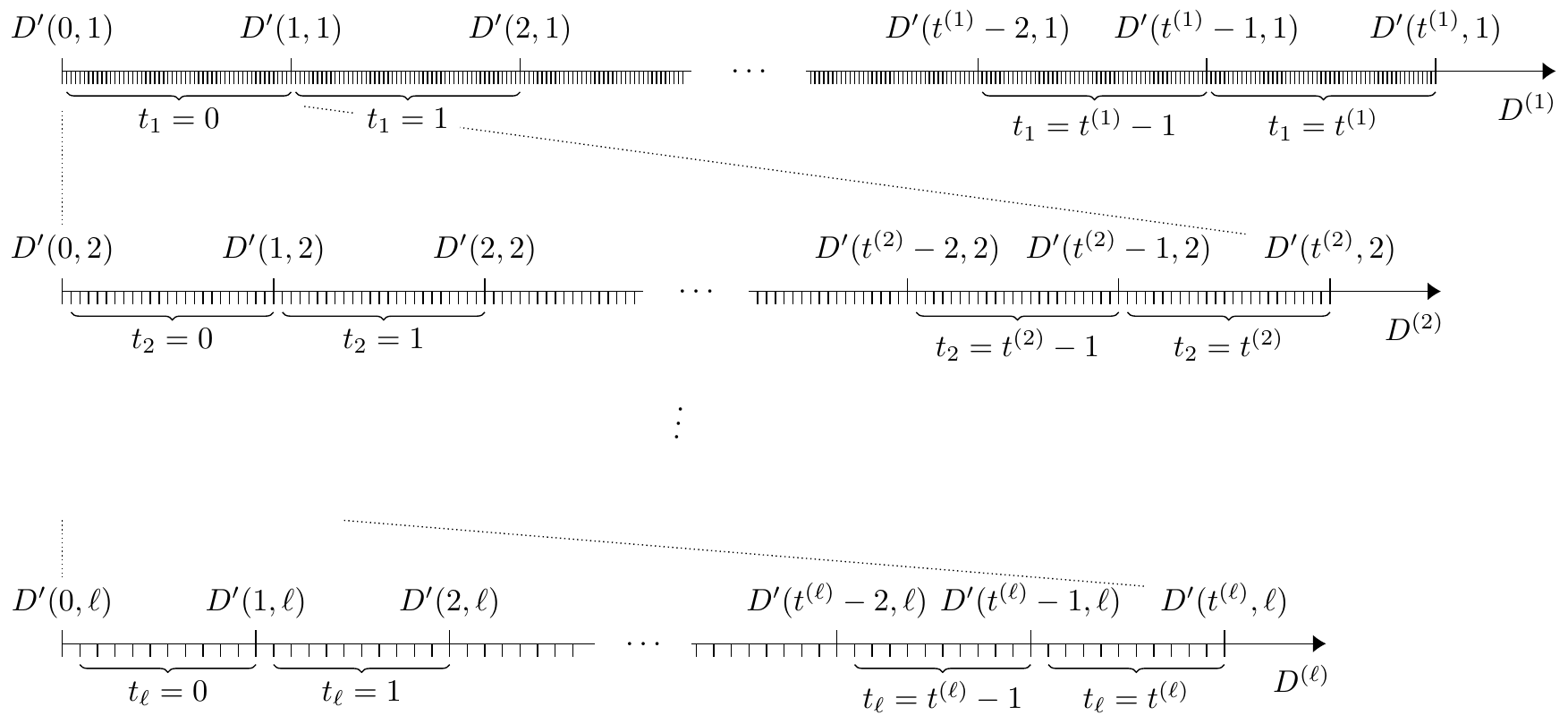}
\caption{\rev{Illustration of the mapping $\phi \circ \varphi  : \{1,\hdots,\mathcal{N}_{q,\npr,m}(t,\ell)\} \rightarrow \Tsett, D^{(1)} \mapsto \t$. The variables are defined as in \cref{alg:uniform_drawing_sum-rank_errors}, and the function $D'(t'',j):=\sum_{t'=t^{(j)}-\nmmin(\ell-j)}^{{t''-1}} \NM{q}{m,\npr,t'} \, \mathcal{N}_{q,\npr,m}({t^{(j)}-t'},\ell-j) $}.}
\label{fig:drawingerrors}
\end{figure}

}

\subsection{Optimal Support-Drawing Algorithm}\label{app:optimal_support_finding_via_linear_programming}

In \cref{sec:super_support_finding}, we saw that the worst-case expected number of iterations of a super-support drawing algorithm that first draws a vector $\s \in \Tsets$ according to a probability distribution $\ps$ and then $\Fset \sample \SupportSet(\s)$, can be given as (cf.~\eqref{eq:worst_case_num_iterations_exact_optimization})
\begin{align*}
\NumItWorstCase = \max_{\t \in \Tsett} \left(\sum_{\s \in \Tsets} \ps \rhoprob(\s,\t)\right)^{-1}.
\end{align*}
\cref{sec:super_support_finding} presented a scalable method to design $\ps$ that can be implemented in polynomial time, but does not guarantee to minimize \eqref{eq:worst_case_num_iterations_exact_optimization}.
Here, we show how to achieve an optimal solution, at the cost of a super-polynomial complexity.
The following theorem reformulates the optimization problem into a linear programming instance.

\begin{theorem}\label{thm:optimal_ps_linear_program}
Fix arbitrary orders $\s_1,\dots,\s_{|\Tsets|}$ and $\t_1,\dots,\t_{|\Tsett|}$ of elements in $\Tsets$ and $\Tsett$, respectively.
Let %
\begin{align*}
\c &= \begin{bmatrix}
0 \\
0 \\
\vdots \\
0 \\
1
\end{bmatrix} \in  \mathbb{R}^{(|\Tsets|+1) \times 1}, \quad \quad
\b = \begin{bmatrix}
0 \\
0 \\
\vdots \\
0 \\
1 \\
-1
\end{bmatrix} \in \mathbb{R}^{|\Tsett| \times 1}, \quad \text{and}  \\
\A &= \begin{bmatrix}
-\rhoprob(\s_1,\t_1) & -\rhoprob(\s_2,\t_1) & \dots & -\rhoprob(\s_{|\Tsets|},\t_1) & 1 \\
-\rhoprob(\s_1,\t_2) & -\rhoprob(\s_2,\t_2) & \dots & -\rhoprob(\s_{|\Tsets|},\t_2) & 1 \\
\vdots & \vdots & \ddots & \vdots & \vdots \\
-\rhoprob(\s_1,\t_{|\Tsett|}) & -\rhoprob(\s_2,\t_{|\Tsett|}) & \dots & -\rhoprob(\s_{|\Tsets|},\t_{|\Tsett|}) & 1 \\
1 & 1 & \dots & 1 & 0 \\
-1 & -1 & \dots & -1 & 0
\end{bmatrix} \in \mathbb{R}^{(|\Tsett|+2) \times (|\Tsets|+1)}.
\end{align*}
If $\x \in \mathbb{R}^{(|\Tsets|+1) \times 1}$ is a solution to the linear program
\begin{align}
&\mathrm{Maximize} &&\c^\top \x \notag \\
&\mathrm{subject \; to} && \A \x \leq \b \label{eq:linear_program} \\
&\mathrm{and} && \x \geq 0, \notag
\end{align}
then $\psnumber = x_i$, for all $i=1,\dots, |\Tsets|$, is a distribution that maximizes \eqref{eq:worst_case_num_iterations_exact_optimization}, and we have
\begin{align}
x_{|\Tsets|+1}^{-1} = \min\!\left\{ \max_{\t \in \Tsett} \left(\sum_{\s \in \Tsets} \ps \rhoprob(\s,\t)\right)^{-1} \, : \, \ps \in [0,1] \, \forall \s \in \Tsets, \, \sum_{\s \in \Tsets} \ps = 1\right\}. \label{eq:x_last_entry_solution}
\end{align}
\end{theorem}

\begin{proof}
We write $\psnumber = x_i$ and $\xi = x_{|\Tsets|+1}$ for a solution $\x$ of the linear program.
The last two rows of $\A$ are equivalent to
\begin{align*}
\sum_{i=1}^{|\Tsets|} \psnumber = 1,
\end{align*}
Together with $\x\geq 0$, we get that the $\psnumber$ form a valid discrete probability mass function.
The first $|\Tsett|$ rows of $\A$ correspond to the constraints
\begin{align*}
\sum_{i=1}^{|\Tsets|} \psnumber \rhoprob(\s_i,\t_j) \geq \xi \quad \forall \, j=1,\dots,|\Tsett|.
\end{align*}
Since $\xi$ is the maximal positive value for which this constaint is fulfilled for all $j=1,\dots,|\Tsett|$ and solutions $\psnumber$,
we have
\begin{align*}
\xi = \max \left\{ \min_{j=1,\dots,|\Tsett|}\!\left\{ \sum_{i=1}^{|\Tsets|} \psnumber \rhoprob(\s_i,\t_j) \right\} \, : \, \psnumber \in [0,1] \, \forall i=1,\dots,|\Tsets|, \sum_{i=1}^{|\Tsets|} \psnumber = 1 \right\}
\end{align*}
which is equivalent to \eqref{eq:x_last_entry_solution}. This proves the claim.
\end{proof}

Using standard methods, the linear program \eqref{eq:linear_program} in \cref{thm:optimal_ps_linear_program} can be solved in polynomial time in the number of variables, $|\Tsets|+1$ (note that the number of linear constraints is in $O(|\Tsets|)$). As, depending on the relative growth of $s$, $\mu$, and $\ell$, this number may grow super-polynomially in $s$, it is usually not possible to solve the linear program efficiently for large code parameters. Furthermore, even if a solution $\x$ is found or pre-computed, it is not apparent how to draw efficiently from the distribution $\psnumber = x_i$ (for all $i=1,\dots, |\Tsets|$).

Nevertheless, we include this ``optimal'' solution to the design of $\ps$ in the discussion in \cref{sec:comparison} for all values of $\ell, \mu,s$ for which we can retrieve a solution in short time (and ignore the issue of efficient drawing).
For these computations, we apply a trick that reduces the number of variables and constraints: We assume that the restriction to those solutions $\x$ such that $x_i = x_j$ for all $i,j$ with permutationally equivalent $\s_i \sim \s_j$. Hence, we can reduce the number of variables to $|\TsetordWITHARGUMENTS{s,\ell,\mu}|+1$ (which may still be super-polynomially in $s$, though) and the number of constraints to $|\Tsettord|+2 \leq |\TsetordWITHARGUMENTS{s,\ell,\mu}|+2$.
The complexity of this generic decoding approach is roughly given by
  \begin{equation*}
   \Woptimal := \Wperiteration \, x_{|\Tsets|+1}^{-1},
  \end{equation*}
where $x_1,\hdots,x_{|\Tsets|+1}$ is a solution vector to the optimization problem in \cref{thm:optimal_ps_linear_program} and $\Wperiteration$ is the cost of one iteration.
The latter value is at least the cost of erasure decoding, which is in $O^\sim(n^3m^3\log(q))$, but the real cost might be larger since we need to be able to efficiently draw from the distribution $\ps$. In the plots in \cref{sec:comparison}, we use the same $\Wperiteration$ as for the other algorithms, which is an optimistic estimate.

It can be seen that in all cases in which we can compute the expected runtime of a generic decoder that draws $\s$ according to such an optimal distribution $\ps$, the ``optimal'' runtime is only insigificantly smaller than the practical solution presented in \cref{sec:super_support_finding}.

\end{document}

%% file: defs.tex
\usepackage{pgfplots}
\pgfplotsset{compat=newest}
\usepackage{tikz}
\usetikzlibrary{arrows,matrix,positioning,patterns,decorations.pathreplacing}
\usetikzlibrary{calc,intersections,through,backgrounds,matrix,patterns}
\usetikzlibrary{arrows.meta}
\usepgfplotslibrary{fillbetween}
\pgfdeclarelayer{bg}
\pgfsetlayers{bg,main}

\tikzset{
    hatch distance/.store in=\hatchdistance,
    hatch distance=100pt,
    hatch thickness/.store in=\hatchthickness,
    hatch thickness=0.3pt
}

\usetikzlibrary{decorations.markings}

\tikzstyle{interrupt}=[
densely dotted,
postaction={
	decorate,
	decoration={markings,
		mark= at position 0.1 
		with
		{
			\fill[white] (-0.5,-0.3) rectangle (0.7,0.3);
		}
	}
}
]

\usepackage[utf8]{inputenc}
\usepackage[english]{babel}
\usepackage[T1]{fontenc}
\usepackage{epsfig}
\usepackage{amsmath, amssymb, amsbsy}
\usepackage{mathdots}
\usepackage{xspace}
\usepackage[noend]{algpseudocode}
\usepackage[linesnumbered,ruled,vlined,titlenumbered]{algorithm2e}
\usepackage{algorithmicx}
\usepackage{color}
\usepackage{cite}
\usepackage{booktabs}
\usepackage{verbatim}
\usepackage[colorlinks=true,citecolor=blue,linkcolor=blue,urlcolor=blue]{hyperref}
\usepackage{url}
\usepackage{lipsum}
\usepackage{enumitem}
\usepackage{colortbl}
\usepackage{amsthm}
\usepackage{dblfloatfix}
\usepackage[capitalise]{cleveref}
\usepackage{nicefrac}

\makeatletter
\newcommand\footnoteref[1]{\protected@xdef\@thefnmark{\ref{#1}}\@footnotemark}
\makeatother

\newtheorem{theorem}{Theorem}
\newtheorem{lemma}[theorem]{Lemma}
\newtheorem{corollary}[theorem]{Corollary}
\newtheorem{proposition}[theorem]{Proposition}

\newtheorem{problem}[theorem]{Problem}

\newtheorem{remark}[theorem]{Remark}

\newtheorem{definition}{Definition}

\usepackage[final,tracking=true,kerning=true,spacing=true,factor=1100,stretch=10,shrink=20]{microtype}

\newenvironment{mymatrix}{\begin{bmatrix}} {\end{bmatrix} }

\def\ve#1{{\mathchoice{\mbox{\boldmath$\displaystyle #1$}}%
              {\mbox{\boldmath$\textstyle #1$}}%
              {\mbox{\boldmath$\scriptstyle #1$}}%
              {\mbox{\boldmath$\scriptscriptstyle #1$}}}}

\newcommand{\rev}[1]{{#1}}

\definecolor{darkgreen}{rgb}{0,0.6,0}

\newcommand{\Fq}{\ensuremath{\mathbb{F}_q}}
\newcommand{\Fqm}{\ensuremath{\mathbb{F}_{q^m}}}

\DeclareMathOperator{\diag}{diag}

\newcommand{\Code}{\mathcal{C}}

\newcommand{\0}{\ve{0}}

\renewcommand{\H}{\ve{H}}

\renewcommand{\a}{\ve{a}}
\renewcommand{\b}{\ve{b}}
\newcommand{\e}{\ve{e}}

\renewcommand{\c}{\ve{c}}
\renewcommand{\d}{\ve{d}}
\renewcommand{\e}{\ve{e}}

\newcommand{\G}{\ve{G}}

\newcommand{\A}{\ve{A}}
\newcommand{\B}{\ve{B}}

\newcommand{\h}{\ve{h}}

\newcommand{\wtR}{\mathrm{wt}_\mathrm{R}}

\newcommand{\ZZ}{\mathbb{Z}}

\newcommand{\rk}{\mathrm{rk}}

\newcommand{\x}{\ve{x}}
\newcommand{\s}{\ve{s}}
\renewcommand{\t}{\ve{t}}
\renewcommand{\r}{\ve{r}}

\newcommand{\Eset}{\mathcal{E}}
\newcommand{\Fset}{\mathcal{F}}
\newcommand{\wtSR}{\mathrm{wt}_{\mathrm{SR},\ell}}
\newcommand{\dSR}{\mathrm{d}_{\mathrm{SR},\ell}}

\newcommand{\qbinomial}[2]{\genfrac{[}{]}{0pt}{}{#1}{#2}_q}
\newcommand{\Iset}{\mathcal{I}}
\newcommand{\wtH}{\mathrm{wt}_\mathrm{H}}
\newcommand{\Jset}{\mathcal{J}}

\newcommand{\NM}[2]{\mathrm{NM}_{#1}(#2)}

\newcommand{\Wcode}{W_{\Code}}
\newcommand{\Werror}{W_{\mathrm{errors}}}
\newcommand{\Wnaive}{W_\mathrm{GRS}}
\newcommand{\Wnew}{W_\mathrm{new}}
\newcommand{\WnewLB}{W_\mathrm{new}^{\mathrm{(LB)}}}
\newcommand{\WnewUB}{W_\mathrm{new}^{\mathrm{(UB)}}}
\newcommand{\WnewUBsimple}{W_\mathrm{new}^{\mathrm{(UB, simple)}}}
\newcommand{\Wperiteration}{W_\mathrm{iter}}
\newcommand{\Woptimal}{W_\mathrm{optimal}}

\newcommand{\idxh}{h}
\newcommand{\idxg}{g}
\newcommand{\idxf}{f}
\newcommand{\srem}{s_{\text{\normalfont rem}}}

\newcommand{\lsr}{\ell}

\newcommand{\npr}{\eta}
\newcommand{\nmmin}{\mu}
\newcommand{\scomp}{\mathsf{scomp}_{\ambdim}}

\newcommand{\DHSD}{\ensuremath{\mathsf{SynDec}_{\mathsf{H}}} \xspace}
\newcommand{\DSRSD}{\ensuremath{\mathsf{SynDec}_{\ell{-}\mathsf{SR}}} \xspace}
\newcommand{\DSRSDnoell}[1]{\ensuremath{\mathsf{SynDec}_{#1{-}\mathsf{SR}}} \xspace }
\newcommand{\ZPP}{\mathsf{ZPP}}
\newcommand{\Pclass}{\mathsf{P}}
\newcommand{\RP}{\mathsf{RP}}
\newcommand{\coRP}{\mathsf{coRP}}
\newcommand{\NP}{\mathsf{NP}}
\newcommand{\PPT}{\mathsf{PPT}}

\newcommand{\tH}{t_{\text{H}}}
\newcommand{\tHtilde}{\tilde{t}_{\text{H}}}

\newcommand{\TsetWITHARGUMENTS}[1]{\mathcal{T}_{#1}}
\newcommand{\Tset}{\TsetWITHARGUMENTS{t,\ell,\nmmin}}
\newcommand{\Tsets}{\TsetWITHARGUMENTS{s,\ell,\nmmin}}
\newcommand{\Tsett}{\Tset}
\newcommand{\Tsettord}{\Tset^{(\mathsf{ord})}}
\newcommand{\TsetordWITHARGUMENTS}[1]{\mathcal{T}_{#1}^{(\mathsf{ord})}}

\newcommand{\Ahco}{\mathcal{A}_{\textrm{H}}^{\coRP}}
\newcommand{\Arco}{\mathcal{A}_{\textrm{SR}}^{\coRP}}

\newcommand{\Ahrp}{\mathcal{A}_{\textrm{H}}^{\RP}}
\newcommand{\Arrp}{\mathcal{A}_{\textrm{SR}}^{\RP}}

\newcommand{\Ap}{\mathcal{A}^{\Pclass}}
\newcommand{\Arp}{\mathcal{A}^{\RP}}
\newcommand{\Acorp}{\mathcal{A}^{\coRP}}
\newcommand{\Azpp}{\mathcal{A}^{\ZPP}}
\newcommand{\Anp}{\mathcal{A}^{\NP}}

\newcommand{\tSR}{{t}_{\textrm{SR}}}
\newcommand{\tSRtilde}{\tilde{t}_{\text{SR}}}

\newcommand{\sample}{\xleftarrow{\$}}

\newcommand{\col}{{\tt KeepCols} }

\newcommand{\SupportSet}{\Xi_{q,\ambdim}}

\newcommand{\Msum}{\mathsf{M}}
\newcommand{\te}{\t_{\e}}
\newcommand{\Esete}{\Eset_{\e}}

\newcommand{\NumItWorstCase}{\max_{\substack{\e \in \Fqm^n \, : \\ \wtSR(\e) = t}} \mathbb{E}[\text{\#iterations}]}

\newcommand{\True}{\mathsf{true}}
\newcommand{\False}{\mathsf{false}}

\usepgfplotslibrary{fillbetween}

\definecolor{constructionAcolor}{rgb}{1,0.4,0}  %
\definecolor{constructionBcolor}{rgb}{0.8,0,1}	%
\definecolor{constructionCcolor}{rgb}{0,0,1}	%
\definecolor{constructionDcolor}{rgb}{0,0,0}	%
\definecolor{constructionEcolor}{rgb}{0,0.7,0}	%

\newcommand{\ambdim}{\zeta}
\newcommand{\ext}{\mathrm{ext}}

\newcommand{\extFqmFqVector}[1]{\ext_{q, m}^{#1}}

\newcommand{\rhoprob}{\varrho_{q,\ambdim}}
\newcommand{\rhoprobt}{\varrho_{q,\ambdim,s}}
\newcommand{\const}{\Delta}

\newcommand{\Q}{Q_{t,\ell,\mu}}

\newcommand{\ps}{\tilde{p}_{\s}}
\newcommand{\psnumber}{\tilde{p}_{\s_i}}